\documentclass[12pt]{article}
\usepackage{amsmath, amsfonts, amssymb, amsthm, fullpage, color, bbm, enumerate, titlesec, graphicx, epstopdf, multirow, array, mathrsfs}
\usepackage[T1]{fontenc}
\usepackage{lmodern}

\usepackage[onehalfspacing]{setspace}

\definecolor{darkred}{rgb}{0.5,0,0}

\titleformat*{\paragraph}{\sc}

\usepackage{hyperref}
\hypersetup{allbordercolors={1 1 1},allcolors=darkred,colorlinks=true}
\usepackage[capitalize,nameinlink,noabbrev]{cleveref}

\usepackage[small,bf]{caption}

\usepackage[round,longnamesfirst]{natbib}
\usepackage[runin]{abstract}

\newcolumntype{C}[1]{>{\centering\arraybackslash}p{#1}}

\theoremstyle{plain}
\newtheorem{asn}{Assumption}
\crefname{asn}{Assumption}{Assumptions}
\newtheorem{lem}{Lemma}
\crefname{lem}{Lemma}{Lemmas}

\newtheorem{prop}{Proposition}
\crefname{prop}{Proposition}{Propositions}

\DeclareMathOperator*{\argmin}{argmin}
\DeclareMathOperator*{\tr}{trace}

\DeclareMathOperator*{\diag}{diag}
\DeclareMathOperator*{\var}{Var}
\DeclareMathOperator*{\cov}{Cov}
\DeclareMathOperator*{\corr}{Corr}

\DeclareMathOperator*{\sgn}{sign}

\allowdisplaybreaks

\newcommand{\alert}[1]{#1} 

\begin{document}

\title{\texorpdfstring{\vspace{-\baselineskip}}{}Standard Errors for Calibrated Parameters\thanks{Email: {\tt matthew.cocci@gmail.com} and {\tt mikkelpm@princeton.edu}. We are grateful for comments from the editor Francesca Molinari, three anonymous referees, Fernando Alvarez, Isaiah Andrews, Kirill Evdokimov, Bo Honor\'{e}, Michal Koles\'{a}r, Elena Manresa, Benjamin Moll, Pepe Montiel Olea, Ulrich M\"{u}ller, and seminar participants at several venues. We thank Minsu Chang and Silvia Miranda-Agrippino for supplying the moments used in our second empirical application. Samya Aboutajdine and Eric Qian provided excellent research assistance. Plagborg-M{\o}ller acknowledges that this material is based upon work supported by the NSF under Grant {\#}1851665. Any opinions, findings, and conclusions or recommendations expressed in this material are those of the authors and do not necessarily reflect the views of the NSF.}}
\author{Matthew D. Cocci \\ Amazon \and Mikkel Plagborg-M{\o}ller \\ Princeton University}
\date{\texorpdfstring{\vspace{0.2\baselineskip}}{}This version: June 17, 2024 \texorpdfstring{\\[0.2ex]}{} First version: June 12, 2019}
\maketitle

\vspace{-\baselineskip}
\begin{abstract}
Calibration, the practice of choosing the parameters of a structural model to match certain empirical moments, can be viewed as minimum distance estimation. Existing standard error formulas for such estimators require a consistent estimate of the correlation structure of the empirical moments, which is often unavailable in practice. Instead, the variances of the individual empirical moments are usually readily estimable. Using only these variances, we derive conservative standard errors and confidence intervals for the structural parameters that are valid even under the worst-case correlation structure. In the over-identified case, we show that the moment weighting scheme that minimizes the worst-case estimator variance amounts to a moment selection problem with a simple solution. Finally, we develop tests of over-identifying or parameter restrictions. We apply our methods empirically to a model of menu cost pricing for multi-product firms and to a heterogeneous agent New Keynesian model.
\end{abstract}

\section{Introduction}
\label{sec:intro}

\alert{The use of structural economic models involves first selecting values for model parameters and, second, evaluating the model's implications for empirical moments or counterfactual quantities under those parameters. This set of tasks is often referred to as ``calibration.'' \cite{Kydland1996} define this term in a macroeconomic context, advocating that one choose the parameters in a non-econometric fashion and then compare simulated moments to empirical counterparts.\footnote{\cite{DeJong2011} and \citet{Nakamura2018} provide further discussion of calibration.} Alternatively, \citet{Hansen1996} and Thomas Sargent (quoted in \citealp[p.\ 568]{Evans2005}) treat ``calibration'' as synonymous with moment matching---or more generally, minimum distance---estimation. We adopt the latter perspective on calibration in this paper, as it provides a unified econometric framework for estimating parameters if desired (provided suitable identification conditions are satisfied), as well as for testing the ability of a structural model to fit empirical moments for any given values of the parameters (whether estimated or fixed as in \citealp{Kydland1996}). Moment matching estimation and specification testing is a widely used paradigm, not only in macroeconomics, but in diverse fields of applied structural economics.}

\alert{The application of moment matching estimators and tests to structural models} requires knowledge of the entire variance-covariance matrix of the empirical moments, but in practice this matrix is often only partially known. When the empirical moments are obtained from different data sets, different econometric methods, or different previous papers, it is usually hard or impossible to estimate the off-diagonal elements of the variance-covariance matrix. Nevertheless, the diagonal of the matrix---the variances of the individual empirical moments---is typically estimable. In this paper, we show that the diagonal suffices to obtain practically useful worst-case standard errors for the moment matching estimator. Moreover, in the over-identified case, we show that the moment weighting that minimizes the worst-case estimator variance amounts to a moment selection problem with a simple solution. \alert{We also propose a test of correct model specification that can be applied both when the parameters are estimated and when they are fixed as in \citet{Kydland1996}.} Hence, our limited-information methods allow researchers to choose their moments and data sources freely without giving up on valid statistical inference.

We show that worst-case standard errors for the structural parameters (or smooth transformations thereof), using only the empirical moment variances, are easy to compute. They are given by a weighted sum of the standard errors of individual empirical moments, where the weights depend on the moment weight matrix and the derivatives of the moments with respect to the structural parameters. The derivatives can be obtained analytically, by automatic differentiation, or by first differences. Using these worst-case standard errors, one can construct a confidence interval that is valid even under the worst-case correlation structure. The confidence interval is generally conservative for specific correlation structures, but its minimax coverage is exact, i.e., under the worst-case correlation structure, which amounts to perfect positive/negative correlation. The confidence interval is likely to be informative in many empirical applications, as it is at most $\sqrt{p}$ times wider than it would be if the moments were known to be independent, where $p$ is the number of moments used for estimation.

Given knowledge of only the individual empirical moment variances, we show that the moment weighting scheme that minimizes the worst-case estimator variance amounts to a moment \emph{selection} problem. That is, the efficient minimum distance weight matrix attaches zero weight to some of the moments. The efficient selection of moments can be conveniently computed by running a median regression (i.e., Least Absolute Deviation regression) on a particular artificial data set. Our limited-information efficient estimator given knowledge of only the moment variances is generally different from the familiar full-information efficient estimator that requires knowledge of the entire moment correlation structure.

To understand the intuition behind our efficiency result, consider the analogy of portfolio selection in finance. This analogy is mathematically relevant, as it is well known that any minimum distance estimator is asymptotically equivalent to a linear combination of the empirical moments---a ``portfolio'' of moments---with a linear restriction on the weights to ensure unbiasedness. When constructing a minimum-variance financial portfolio that achieves a given expected return, it is usually optimal to diversify across all available assets, \emph{except} if the assets are perfectly (positively or negatively) correlated. In the latter extreme case, it is optimal to entirely disregard assets with sufficiently high variance relative to their expected return. But it is precisely the extreme case of perfect correlation that delivers the worst-case variance of a given portfolio. Thus, the portfolio with the smallest worst-case variance across correlation structures is a portfolio that \emph{selects} a subset of the available assets.

We propose joint tests of parameter restrictions as well as tests of over-identifying restrictions. A common form of over-identification test used in the empirical literature is to check whether the estimated structural parameters yield a good fit of the model to ``non-targeted'' moments, i.e., moments that were not exploited for parameter estimation. We show how to implement a formal statistical test based on this idea in our setup. \alert{The test can be performed either using estimated parameters (if these are identified from the targeted moments) or given pre-specified parameter values (perhaps obtained from previous studies) as in the approach advocated by \citet{Kydland1996}. Separately, we develop a Wald-type joint test of parameter restrictions.} The proof of the validity of this test relies on tail probability bounds for quadratic forms in Gaussian vectors from \citet{Szekely2003}, but the test statistic and critical value are simple and easy to compute.

Finally, we extend our limited-information procedures to settings with more detailed knowledge of the covariance matrix of empirical moments. This includes settings where the entire correlation structure is known for some subsets of the moments, or where certain moments are known to be independent of each other.

We illustrate the usefulness of our procedures through two empirical applications. In the first one, we estimate and test the \citet{Alvarez2014} model of menu cost price setting in multi-product firms, by matching moments of price changes. In the second application, we estimate and test a heterogeneous agent New Keynesian model developed by \citet{McKay2016} and \citet{Auclert2020}, by matching impulse responses for macro time series and cross-sectional micro moments. Our worst-case standard errors yield informative inference on several parameters of interest in both applications. A Monte Carlo simulation study calibrated to the first application indicates that our methods perform well in finite samples.

\paragraph{Literature.}
Unlike the literature on correlation matrix completion \citep[e.g.,][]{Georgescu2018}, we solve the explicit problem of finding the worst-case correlation structure when estimating parameters in a structural model. Our derivations of the worst-case efficient weight matrix and joint testing procedure do not seem to have parallels in the matrix completion literature.

While we focus on cases where it is difficult to estimate the correlation structure of different moments, in some applications it may be possible to model and exploit the precise relationship between the moments. The literature on estimating heterogeneous agent models in macroeconomics has recently developed procedures for combining macro and micro data, as discussed further in \cref{sec:appl_hank}. \citet{Hahn2022prod,Hahn2022mixingale} provide advanced tools for doing inference with a mix of cross-sectional and time series data. These methods, unlike ours, generally require access to the underlying data. \citet{Imbens1994} consider a microeconometric setting where certain macro moments are known without error, which is a special case of our framework. \citet{Hahn2020Estimation} give examples of structural models where both time series and cross-sectional data are required for identification of structural parameters. Their insights may help inform the choice of moments for the methods that we develop below.

\paragraph{Outline.}
\cref{sec:model} defines the moment matching setup. \cref{sec:procedure} derives the worst-case standard errors and the efficient moment weighting/selection. \cref{sec:test} develops tests of parameter restrictions and of over-identifying restrictions. \cref{sec:appl} contains two empirical illustrations. \cref{sec:concl} concludes. \cref{sec:appendix} contains proofs and other technical details. Code for implementing our procedures is available online.

\section{Setup} \label{sec:model}

Consider a standard moment matching (minimum distance) estimation framework. Let $\mu_0 \in \mathbb{R}^p$ be a vector of reduced-form parameters, which we will refer to as ``moments'', though the method applies more generally. Let $\theta_0 \in \Theta \subset \mathbb{R}^k$ be a vector of structural model parameters. According to an economic model, the two parameter vectors are linked by the relationship $\mu_0 = h(\theta_0)$, where $h \colon \Theta \to \mathbb{R}^p$ is a known function implied by the model. The map $h(\cdot)$ may be computed either analytically or numerically. We have access to an estimator $\hat{\mu}$ (``empirical moments'') that satisfies
\begin{equation} \label{eqn:mu_asyn}
\sqrt{n}(\hat{\mu}-\mu_0) \stackrel{d}{\to} N(0,V)
\end{equation}
for a $p\times p$ symmetric positive semidefinite variance-covariance matrix $V$.\footnote{Here and below, all limits are taken as the sample size $n\to\infty$. We implicitly think of the sample sizes for the different moments as being proportional, with the factors of proportionality reflected in $V$. If some element $\hat{\mu}_j$ converges at a faster rate than $\sqrt{n}$, then $V_{jj}=0$. Sample sizes and convergence rates only enter into our practical procedures through their implicit effect on the calculation of the moment standard errors $\hat{\sigma}_j$ (discussed below), which is handled automatically by econometric software.} Let $\hat{W}$ be a $p \times p$ symmetric matrix satisfying $\hat{W} \stackrel{p}{\to} W$ (we discuss the choice of $\hat{W}$ in \cref{sec:weight}). Then a ``moment matching'' or  ``minimum distance'' estimator of $\theta_0$ is given by
\begin{equation} \label{eqn:estimator}
\hat{\theta} = \argmin_{\theta \in \Theta}\; (\hat{\mu}-h(\theta))'\hat{W}(\hat{\mu}-h(\theta)).
\end{equation}
This estimation strategy is sometimes referred to as ``calibration''.\footnote{Our analysis extends in a straight-forward manner to Generalized Minimum Distance estimation. In that setting $\theta_0$ and $\mu_0$ are linked through a possibly non-separable equation $g(\theta_0,\mu_0)=0_{m \times 1}$. The setting in our paper is a special case with $g(\theta,\mu) = h(\theta)-\mu$, but our calculations carry over with few changes because the asymptotic expansions are essentially the same \citep{Newey1994}.}

If we were able to estimate the covariance matrix of the empirical moments $\hat{\mu}$ consistently, it would be straight-forward to construct standard errors for any smooth function of the estimator $\hat{\theta}$. Suppose we are interested in the scalar transformed parameter $r(\theta_0)$, where $r \colon \Theta \to \mathbb{R}$. For example, $r(\cdot)$ may equal a particular counterfactual quantity in the structural model, or we could simply set $r(\theta)=\theta_i$ for some index $i$. Under the standard regularity conditions listed below in \cref{asn:regularity},
\begin{align}
\sqrt{n}(r(\hat{\theta})-r(\theta_0)) &= \lambda'(G'WG)^{-1}G'W \sqrt{n}(\hat{\mu}-\mu_0) + o_p(1) \label{eqn:delta_method} \\
&\stackrel{d}{\to} N\left(0, \lambda'(G'WG)^{-1}G'WVWG(G'WG)^{-1}\lambda\right), \nonumber
\end{align}
where $G \equiv \partial h(\theta_0)/\partial \theta' \in \mathbb{R}^{p \times k}$ and $\lambda \equiv \partial r(\theta_0)/\partial \theta \in \mathbb{R}^k$. See \citet{Newey1994} for details. If the entire asymptotic covariance matrix $V$ of $\hat{\mu}$ were consistently estimable, the above display would allow computation of standard errors, confidence intervals, and hypothesis tests.

Unfortunately, the full correlation structure of $\hat{\mu}$ is difficult or impossible to estimate in certain applications. This may be the case, for example, when the moments $\hat{\mu}$ are obtained from a variety of different data sources or econometric methods, or from previous studies for which the underlying data is not readily available. Moreover, if the moments involve a mix of time series and cross-sectional data sources, it can be difficult conceptually or practically to estimate correlations across data sources, whether through the bootstrap or the Generalized Method of Moments (GMM). While the structural model could in some cases be exploited to estimate the moment covariance matrix, this may require stronger assumptions than what is needed for point estimation of the structural parameters.\footnote{For example, if we exploit the model's predictions about second moments for estimation, model-based estimation of $V$ would require believing the model's predictions about fourth moments.} We discuss these points in more detail at the end of this subsection and in the empirical applications in \cref{sec:appl}.

Yet, it is often the case that the standard errors of each of the components of $\hat{\mu}$ are available. These marginal standard errors may be directly computable from data, or they may be reported in the various papers that the individual elements of $\hat{\mu}$ are gleaned from. Thus, assume that we have access to standard errors $\hat{\sigma}_1,\dots,\hat{\sigma}_p \geq 0$ satisfying
\begin{equation} \label{eqn:se_consist}
\sqrt{n}\hat{\sigma}_j \stackrel{p}{\to} V_{jj}^{1/2},\quad j=1,\dots,p.
\end{equation}
We show in the next section that these marginal standard errors suffice for doing informative inference on $r(\theta_0)$.\footnote{Our set-up covers applications where the parameters $\hat{\theta}$ are not explicitly estimated but are instead fixed at certain \alert{extraneous} values $\hat{\mu}$ obtained from prior studies. In this case $h(\theta)=\theta$, and $\hat{\sigma}_j$ is the standard error (which could potentially be 0 for some $j$) of the $j$-th \alert{extraneous} parameter value $\hat{\theta}_j=\hat{\mu}_j$.}

For ease of reference, we summarize our technical assumptions here:
\begin{asn} \label{asn:regularity}
\leavevmode
\begin{enumerate}[i)]
\item The empirical moment vector $\hat{\mu}$ is asymptotically normal, as in \eqref{eqn:mu_asyn}. \label{itm:asn_regularity_1}
\item The standard errors $\hat{\sigma}_j$ are consistent, as in \eqref{eqn:se_consist}. \label{itm:asn_regularity_2}
\item $h(\cdot)$ and $r(\cdot)$ are both continuously differentiable in a neighborhood of $\theta_0$ (which lies in the interior of $\Theta$), $G \equiv \partial h(\theta_0)/\partial \theta'$ has full column rank, and $\lambda \equiv \partial r(\theta_0)/\partial \theta \neq 0_{k \times 1}$. \label{itm:asn_regularity_3}
\item $\hat{\theta} \stackrel{p}{\to} \theta_0$. \label{itm:asn_regularity_4}
\item $\hat{W} \stackrel{p}{\to} W$ for a symmetric positive semidefinite matrix $W$.
\item $G'WG$ is nonsingular. \label{itm:asn_regularity_6}
\end{enumerate}
\end{asn}
\noindent Conditions (\ref{itm:asn_regularity_2})--(\ref{itm:asn_regularity_6}) are standard regularity conditions that are satisfied in smooth, identified models \citep{Newey1994}. Note that we allow for the possibility that some moments are known with certainty ($V_{jj}=0$) as in \citet{Imbens1994}. We will now discuss the key condition (\ref{itm:asn_regularity_1}).

\paragraph{Discussion of limited information about the correlation structure and the joint normality assumption.} We now provide several examples of situations where it is difficult or impossible to estimate the correlation structure of different types of empirical moments. In cases when the elements of the moment vector $\hat{\mu}$ are obtained from different data sets, the joint normality assumption \eqref{eqn:mu_asyn} requires justification. This ensures not only that a normal distribution is the appropriate reference distribution for obtaining critical values, but also that the vector $\hat{\mu}$ can reasonably be viewed as arising from some joint, repeatable experiment for which the given standard errors $\hat{\sigma}_j$ capture all sources of uncertainty. The joint normality assumption is most easily understood and justified under a model-based (e.g., shock-based) perspective on uncertainty. In this view, there exists a coherent data generating process with both aggregate and idiosyncratic shocks that affect all of the observed data. The empirical application of \cref{sec:appl_hank} is a prototypical example of this framework, but such applications are not the only use case.

There are several cases in which the joint normality assumption appears reasonable, but estimation of the full moment covariance matrix $V$ could be challenging. For example:
\begin{enumerate}
\item \emph{The moments are obtained from the same or similar data sets, but the underlying data for some of the moments is not available.} For example, some moments may be reported in previous papers that use proprietary data. In this case, traditional full-information inference procedures are inapplicable because it is generally impossible to estimate moment cross-correlations without access to the underlying data. However, as long as we have access to marginal standard errors for each individual moment, our methods can be applied. See \cref{sec:appl_price} for an empirical example.
\item \emph{Some of the moments are computed from aggregate time series and others from panel data spanning similar time periods.} If the clustering procedure of the panel data regressions allows for aggregate shocks (such as when clustering by time period or when using the standard error formula from \citealp{Driscoll1998}), and these aggregate shocks also affect the time series data, then the panel regressions will have correct standard errors but the coefficients may be correlated with the time series moments.
\item \emph{The moments stem from time series data observed at various frequencies, or from regional data with various levels of geographic aggregation.} While careful econometric analysis may allow the estimation of the full covariance matrix of the moments by appropriately collapsing the sample moment conditions for the higher-frequency or higher-resolution data, this could be cumbersome in practice. This is especially true if there is limited overlap in the time spans of the higher-frequency and lower-frequency data, or in the geographical coverage of the higher-resolution and lower-resolution data.
\item \emph{We use a combination of aggregate time series moments and micro moments from surveys, and the latter measure time-invariant parameters that are not affected by macro shocks in the sample.} In this case, it is often reasonable to assume that the uncertainty in the micro moments (arising purely from idiosyncratic noise) is uncorrelated with the uncertainty in the macro moments. Similarly, cross-sectional moments obtained from different independent random samples are plausibly uncorrelated with each other. Such extra information about off-diagonal entries of $V$ can be incorporated in our procedures, as shown in \cref{sec:ext}.
\item \emph{The moments are all computed from the same data set, but using a variety of complicated procedures.} While it may in principle be possible to estimate the correlation structure analytically by stacking the first-order conditions for the different estimation procedures into a large GMM moment vector, applied researchers may wish to avoid these complications. The bootstrap offers a simpler approach but is computationally impractical in some cases. In contrast, our analytical inference procedures below only require researchers to obtain standard errors for each estimated moment separately, and these are often outputted automatically by existing software. One may worry about the internal consistency of mixing moments arising from different estimation procedures, in the sense that the models used to compute the various moments could be mutually incompatible (or incompatible with the structural model that the researcher seeks to fit to the moments). This can be tested using the over-identification test that we propose in \cref{sec:test_overid}.
\item \emph{The moments are all computed from the same data set, but
  full-information inference procedures perform poorly in finite
samples.} \alert{In this case full-information efficient inference is feasible and asymptotically valid. However, as shown by \citet{Altonji1996}, traditional full-information
efficient minimum distance estimation, which relies on an estimate
$\hat{V}$ of the entire moment variance-covariance matrix $V$, can be
subject to severe finite-sample biases in many applications, especially
when $p$ is large.}\footnote{When matching simple moments with access to
the underlying i.i.d.\ data, generalized empirical likelihood estimators
are known to have better finite-sample properties than efficient minimum
distance or GMM, though they are often harder to compute \citep[see][for
a review]{Imbens2002}.}
\alert{Though we do not have formal results characterizing finite-sample
performance, one could apply our limited-information inference
procedure as a practical \emph{conservative} robustness check that avoids the noise arising from
estimating the non-diagonal entries of $V$.}
\end{enumerate}
To be clear, in cases 2--6 above (though not in case 1), full-information inference may be possible in principle, as we have indicated. However, this could be cumbersome to implement in practice and/or suffer from poor finite-sample properties. The limited-information procedures we develop offer applied researchers simpler tools that in many cases provide informative inference about structural parameters. Moreover, the limited-information results could be used as a first step for exploring whether it is worthwhile to expend the extra effort required for full-information inference.

It is important to note that the joint normality assumption \eqref{eqn:mu_asyn} may fail in some cases, rendering our methods invalid. For example:
\begin{enumerate}
\item \emph{We use a combination of aggregate time series moments and micro moments from surveys, but the latter are affected by aggregate macro shocks that shift the whole micro outcome distribution.} In this case, standard cross-sectional moments may not even be \emph{consistent} for the true underlying population moments, since the aggregate shocks do not get averaged out \citep[Section 3]{Hahn2020Estimation}. Moreover, the usual micro standard errors will not take into account the \emph{combined} uncertainty in the macro shocks and idiosyncratic micro noise. However, if the effects of macro shocks on micro moments do average out and we appropriately account for all types of uncertainty when computing moments and their standard errors, then our methods below can be applied. We illustrate this empirically in \cref{sec:appl_hank}.
\item \emph{The data used to compute some of the moments is very heavy tailed, or the estimation procedures are not asymptotically regular.} In this case, even \emph{marginal} normality of the individual empirical moments may fail. We leave extensions to non-normal limit distributions as an interesting topic for future research.
\end{enumerate}

\section{Standard errors and moment selection} \label{sec:procedure}
We first derive the worst-case standard errors for a given choice of moment weight matrix. Then we show that the weighting scheme that minimizes the worst-case standard errors amounts to a moment selection problem with a simple solution.

\subsection{Worst-case standard errors and confidence intervals} \label{sec:se}
We first compute the worst-case bound on the standard error of the moment matching estimator, given knowledge of only the variances of the empirical moments. Although the mathematical argument is straight-forward, it appears that the literature has not realized the practical utility of this result.

Recall that we seek to do inference on the scalar parameter $r(\theta_0)$. By the standard delta method expansion \eqref{eqn:delta_method} under \cref{asn:regularity}, the estimator $r(\hat{\theta})$ is asymptotically equivalent to a certain linear function $x'\hat{\mu}$ of the empirical moments, where $x=(x_1,\dots,x_p)' \equiv WG(G'WG)^{-1}\lambda$. We thus seek to bound the variance of an (asymptotically) known linear combination of $\hat{\mu}$, knowing the variance of each component $\hat{\mu}_j$ but not the correlation structure. This worst-case variance is attained when all components of $\hat{\mu}$ are perfectly positively or negatively correlated (depending on the signs of the elements of $x$), yielding the worst-case variance $(\sum_{j=1}^p |x_j| \var(\hat{\mu}_j)^{1/2})^2$, as is proved in the following elementary lemma.\footnote{The basic intuition is that $\var(X+Y) = \var(X)+\var(Y)+2[\var(X)\var(Y)]^{1/2}\corr(X,Y) \leq \var(X)+\var(Y)+2[\var(X)\var(Y)]^{1/2} = (\var(X)^{1/2}+\var(Y)^{1/2})^2$, since $\corr(X,Y) \leq 1$.}

\begin{lem} \label{thm:var_bound}
	Let $x=(x_1,\dots,x_p)' \in \mathbb{R}^p$ and $\sigma_1^2,\dots,\sigma_p^2 \geq 0$. Let $\mathcal{S}(\sigma)$ denote the set of $p \times p$ symmetric positive semidefinite matrices with diagonal elements $\sigma=(\sigma_1^2,\dots,\sigma_p^2)'$. Then
	\[\max_{V \in \mathcal{S}(\sigma)} \sqrt{x'Vx} = \sum_{j=1}^p \sigma_j |x_j|.\]
\end{lem}
\begin{proof}
	The right-hand side is attained by $V = ss'$, where $s=(\sigma_1 \sgn(x_1),\dots,\sigma_p \sgn(x_p))'$. Moreover, for any $V \in \mathcal{S}(\sigma)$,
	\[x'Vx = \sum_{j=1}^p\sum_{\ell=1}^p x_jx_\ell V_{j\ell} \leq \sum_{j=1}^p\sum_{\ell=1}^p |x_jx_\ell|\, |V_{j\ell}|\leq \sum_{j=1}^p\sum_{\ell=1}^p |x_jx_\ell|\, \sigma_j\sigma_\ell = \left(\sum_{j=1}^p \sigma_j |x_j|\right)^2,\]
	where the penultimate inequality uses that $|V_{j\ell}|^2 \leq V_{jj}V_{\ell\ell}$ for any symmetric positive semidefinite matrix $V$.
\end{proof}

We can thus construct an estimate of the worst-case standard error of $r(\hat{\theta})$ as
\begin{equation} \label{eqn:wcse}
\widehat{\text{se}}(\hat{x}) \equiv \sum_{j=1}^p \hat{\sigma}_j|\hat{x}_j|,
\end{equation}
where $\hat{x} = (\hat{x}_1,\dots,\hat{x}_p)' \equiv \hat{W}\hat{G}(\hat{G}'\hat{W}\hat{G})^{-1}\hat{\lambda}$, $\hat{G} \equiv \frac{\partial h(\hat{\theta})}{\partial \theta'}$, and $\hat{\lambda} \equiv \frac{\partial r(\hat{\theta})}{\partial \theta}$. In practice, the derivatives may be computed analytically, by automatic differentiation, or by finite differences. Let $\Phi(\cdot)$ denote the standard normal distribution function. Then the confidence interval
\[\widehat{\text{CI}} \equiv \left[r(\hat{\theta}) - \Phi^{-1}(1-\alpha/2)\widehat{\text{se}}(\hat{x}), r(\hat{\theta}) + \Phi^{-1}(1-\alpha/2)\widehat{\text{se}}(\hat{x}) \right]\]
covers $r(\theta_0)$ with probability at least $1-\alpha$ asymptotically, as shown formally in the following proposition. The asymptotic coverage probability generally exceeds $1-\alpha$ due to the worst-case perspective, but coverage is exact in a particular special case when all elements of the empirical moment vector $\hat{\mu}$ are perfectly correlated asymptotically.
\begin{prop} \label{thm:coverage}
Impose \cref{asn:regularity} and $\max_j |x_j|V_{jj}>0$. Then $\lim_{n \to \infty}P(r(\theta_0) \in \widehat{\text{CI}}) \geq 1-\alpha$, and the inequality binds for the rank-1 matrix $V$ defined in the proof of \cref{thm:var_bound}.
\end{prop}
\begin{proof}
Please see \cref{sec:appendix_proofs}.
\end{proof}

\paragraph{Remarks.}
\begin{enumerate}
\item The worst-case standard errors are at most $\sqrt{p}$ times larger than the standard errors that assume all elements of $\hat{\mu}$ to be mutually uncorrelated. This follows from H\"{o}lder's inequality between $\ell_1$ and $\ell_2$ norms, which yields $\sum_{j=1}^p \hat{\sigma}_j|\hat{x}_j| \leq p^{1/2}(\sum_{j=1}^p \hat{\sigma}_j^2\hat{x}_j^2)^{1/2}$.
\item Though often informative, limited-information inference can potentially have much lower power than full-information inference. Specifically, the worst-case standard errors can in some (but not all) models be arbitrarily larger than the standard errors one would report given full knowledge of $V$. For example, consider a ``repeated measurements'' model with $k=1$, $p=2$, $h(\theta)=(\theta,\theta)'$, $r(\theta)=\theta$, and $\hat{W}=I$, yielding the estimator $\hat{\theta}=(\hat{\mu}_1+\hat{\mu}_2)/2$. The worst-case standard error of $\hat{\theta}$ equals $(\hat{\sigma}_1+\hat{\sigma}_2)/2$, corresponding to the case where $\hat{\mu}_1$ and $\hat{\mu}_2$ are perfectly positively correlated. But it is easy to verify that the \emph{best-case} standard error equals $|\hat{\sigma}_1-\hat{\sigma}_2|/2$, corresponding to $\hat{\mu}_1$ and $\hat{\mu}_2$ being perfectly negatively correlated. If $\hat{\sigma}_1=\hat{\sigma}_2$, the ratio of worst-case to best-case standard errors is infinite.
\item In applications where the moment function $h(\theta)$ cannot be evaluated analytically, the moments can instead be approximated by simulating from the economic model given parameter vector $\theta$. \cref{thm:coverage_sim} in \cref{sec:appendix_sim} shows that the conclusions of \cref{thm:coverage} continue to apply in this case, as long as the number of simulation draws is sufficiently large relative to the empirical sample size $n$.
\end{enumerate}

\subsection{Efficient moment selection} \label{sec:weight}
We now derive a weight matrix $W$ that minimizes the \emph{worst-case} variance of the estimator, derived above. We show that this weight matrix puts weight on at most $k$ moments, so the procedure amounts to efficient moment \emph{selection}. Since the weight matrix $W$ only matters in the over-identified case, we assume $p>k$ in this section.

We seek a weight matrix $W$ that minimizes the worst-case asymptotic standard deviation of $r(\hat{\theta})$. Let $x(W)$ denote the vector $x$ defined in \cref{sec:se}, viewed as a function of $W$, and let $\mathcal{S}_p$ denote the set of $p\times p$ symmetric positive semidefinite matrices $W$ such that $G'WG$ is nonsingular. Then we solve the problem
\begin{align}
&\min_{W \in \mathcal{S}_p}\max_{\tilde{V} \in \mathcal{S}(\diag(V))} \sqrt{\lambda' (G'WG)^{-1}G'W\tilde{V}WG(G'WG)^{-1}\lambda} \label{eqn:opt_weight_problem} \\
&\qquad = \min_{W \in \mathcal{S}_p} \max_{\tilde{V} \in \mathcal{S}(\diag(V))}  (x(W)'\tilde{V}x(W))^{1/2} \nonumber \\
&\qquad = \min_{W \in \mathcal{S}_p} \;\sum_{j=1}^p V_{jj}^{1/2}|x_j(W)|, \nonumber
\end{align}
where the last equality uses \cref{thm:var_bound} in \cref{sec:se}. \cref{thm:onestep} in \cref{sec:appendix_proofs} shows that the solution to the final optimization problem above is given by
\begin{equation} \label{eqn:lad}
\min_{W \in \mathcal{S}_p} \;\sum_{j=1}^p V_{jj}^{1/2}|x_j(W)| = \min_{z \in \mathbb{R}^{p-k}}\; \sum_{j=1}^p |\tilde{Y}_j - \tilde{X}_j'z|,
\end{equation}
where we define
\[\tilde{Y}_j \equiv V_{jj}^{1/2}G_{j\bullet}(G'G)^{-1}\lambda \in \mathbb{R},\quad \tilde{X}_j \equiv -V_{jj}^{1/2}G_{j\bullet}^{\perp \prime} \in \mathbb{R}^{p-k},\]
$G^\perp$ is any $p \times (p-k)$ matrix with full column rank satisfying $G'G^\perp = 0_{k \times (p-k)}$, and the notation $A_{j\bullet}$ means the $j$-th row of matrix $A$. The intuition for the equality \eqref{eqn:lad} is that the set of all minimum distance estimators for various weight matrices is asymptotically equivalent to the set of all estimators that are linear combinations of the $p$ moments $\hat{\mu}$, subject to $k$ asymptotic unbiasedness constraints. We can therefore optimize over a $(p-k)$-dimensional linear space.

The final optimization problem \eqref{eqn:lad} is a median regression (Least Absolute Deviation regression) of the artificial ``regressand'' $\lbrace \tilde{Y}_j \rbrace$ on the $p-k$ artificial ``regressors'' $\lbrace \tilde{X}_j \rbrace$. This regression can be executed efficiently using standard quantile regression software.

The solution to the median regression amounts to optimally selecting at most $k$ of the $p$ moments for estimation. Theorem 3.1 of \citet{Koenker1978} implies that there exists a solution $z^*$ to the median regression \eqref{eqn:lad} such that at least $p-k$ out of the $p$ median regression residuals
\[e_j^* \equiv \tilde{Y}_j - \tilde{X}_j'z^*, \quad j=1,\dots,p,\]
equal zero. Hence, an efficient weight matrix $W^*$ that achieves the minimum in \eqref{eqn:lad} will yield a linear combination vector $x(W^*) = (V_{11}^{-1/2}e_1^*,\dots,V_{pp}^{-1/2}e_p^*)'$ that attaches nonzero weight to at most $k$ out of the $p$ empirical moments $\hat{\mu}$. In other words, the solution to the efficient moment \emph{weighting} problem is endogenously achieved by an efficient moment \emph{selection}. We may pick an arbitrary weight matrix that attaches nonzero weight to only the efficiently selected moments (the magnitudes of the weights do not matter asymptotically, as the selected set of moments is just-identified).

\paragraph{Analytical illustration.}
To illustrate the above results, and to link back to the intuitive discussion of portfolio selection in \cref{sec:intro}, consider a ``repeated measurements'' model with $k=1$, $p=2$, $h(\theta)=(\theta,\theta)'$, and $r(\theta)=\theta$. It is easy to see that any weight matrix $\hat{W}$ yields a minimum distance estimator of the form $\hat{\theta}=x_1\hat{\mu}_1+x_2\hat{\mu}_2$, where $x_1+x_2=1$. The constraint that $x_1$ and $x_2$ sum to 1 ensures that the estimator is consistent. The choice of optimal weight matrix corresponds to an optimal choice of $x_1$ and $x_2=1-x_1$. One might expect the usual diversification motive to cause us to attach nonzero weights to both empirical moments $\hat{\mu}_1$ and $\hat{\mu}_2$, so that their estimation errors average out partially. However, for the worst-case correlation structure for these moments---perfect correlation---their estimation errors amplify each other, rendering diversification futile. Hence, we minimize the worst-case standard error $|x_1|\hat{\sigma}_1+|1-x_1|\hat{\sigma}_2$ by setting $x_1=1$ when $\hat{\sigma}_1<\hat{\sigma}_2$ and otherwise $x_1=0$, yielding the estimator $\hat{\theta}=\hat{\mu}_{\argmin_j \hat{\sigma}_j}$. That is, the efficient estimator simply discards the moment with the largest standard error. In more general models, moment selection also depends on how informative each moment is about the structural parameters, as measured by the Jacobian matrix $G$.

\paragraph{Algorithm.}
The efficient estimator and standard errors can be computed as follows:
\begin{enumerate}[i)]
\item Compute an initial $\sqrt{n}$-consistent estimator $\hat{\theta}_\text{init}$ using, say, a diagonal weight matrix with $\hat{W}_{jj} = \hat{\sigma}_j^{-2}$.
\item Construct the derivative matrix $\hat{G} \equiv \frac{\partial h(\hat{\theta}_\text{init})}{\partial \theta'}$ and vector $\hat{\lambda} \equiv \frac{\partial r(\hat{\theta}_\text{init})}{\partial \theta}$, either analytically or numerically.
\item Solve the median regression \eqref{eqn:lad}, substituting $\hat{G}$ for $G$, $\hat{\lambda}$ for $\lambda$, and $\hat{\sigma}_j$ for $V_{jj}^{1/2}$.\footnote{Remember to omit an intercept from the regression.} Compute the residuals $\hat{e}_j^*$, $j=1,\dots,p$, from this median regression. (In the non-generic case where multiple solutions to the median regression exist, select one that yields at most $k$ nonzero residuals.)
\item Construct the efficient linear combination $\hat{x}^*=(\hat{x}_1^*,\dots,\hat{x}_p^*)'$ of the $p$ moments, given by $\hat{x}_j^*\equiv\hat{\sigma}_{j}^{-1}\hat{e}_j^*$ for $j=1,\dots,p$. At least $p-k$ of the elements will be zero, corresponding to those moments that are discarded by the efficient moment selection procedure.
\item To compute an efficient estimator of $r(\theta_0)$, either:
\begin{enumerate}[a)]
\item Compute the just-identified efficient minimum distance estimator $\hat{\theta}_\text{eff}$ of $\theta_0$ which uses any weight matrix that attaches zero weight to those (at least) $p-k$ moments which receive zero weight in the vector $\hat{x}^*$. Then estimate $r(\theta_0)$ by $r(\hat{\theta}_\text{eff})$. Or:
\item Compute the ``one-step'' estimator $\hat{r}_\text{eff-1S} \equiv r(\hat{\theta}_\text{init}) + \hat{x}^*{}'(\hat{\mu}-h(\hat{\theta}_\text{init}))$ of $r(\theta_0)$.\footnote{See \citet[Section 3.4]{Newey1994} for a general discussion of one-step estimators.}
\end{enumerate}
\label{itm:algo_estim}
\item The worst-case standard error of the estimator from step (\ref{itm:algo_estim}) is given by the value of the median regression \eqref{eqn:lad} (i.e., the minimized objective function).
\end{enumerate}
Options (a) and (b) in step (\ref{itm:algo_estim}) of the algorithm are asymptotically equivalent. Option (b) is computationally more convenient as it avoids further numerical optimization, but option (a) ensures that $\hat{\theta}_\text{eff}$ always lies in the parameter space $\Theta$.

\paragraph{Remarks.}
\begin{enumerate}

\item Since all operations involved in computing the efficient linear combination $\hat{x}^*$ are continuous, $\hat{x}^*$ converges in probability to the population efficient linear combination $x(W^*)$. The only exception may be where the population median regression \eqref{eqn:lad} does not have a unique minimum, which is a non-generic case. Even in this case, however, the efficient worst-case standard errors will be consistent (when multiplied by $\sqrt{n}$) under \cref{asn:regularity}(\ref{itm:asn_regularity_1})--(\ref{itm:asn_regularity_4}), by a standard application of the maximum theorem.

\item The full-information (infeasible) efficient weight matrix that exploits knowledge of all of $V$ is known to equal $W=V^{-1}$. This weight matrix in general attaches nonzero weight to all moments, unlike the limited-information efficient solution derived above. The worst-case asymptotic standard deviation \eqref{eqn:lad} given limited information is of course larger than the asymptotic standard deviation $(\lambda'(G'V^{-1}G)^{-1}\lambda)^{1/2}$ of the full-information efficient estimator of $r(\theta_0)$.

\item The efficient moment weighting/selection for the limited-information efficient estimators $r(\hat{\theta}_\text{eff})$ and $\hat{r}_\text{eff-1S}$ depends on the function $r(\cdot)$ of interest, unlike in the case of full-information efficient estimation. In practice, we can just re-run the computations for each function $r(\cdot)$ of interest (e.g., for each component of $\theta$).

\item By restricting $W$ to the set $\mathcal{S}_p$ of weight matrices for which $G'WG$ is nonsingular, we ensure that the parameter vector $\theta_0$ is at least \emph{locally} identified from the selected moments \citep[p.\ 2144]{Newey1994}. However, the selected moments may fail to ensure \emph{global} identification of $\theta_0$. In models for which this is a concern, we recommend using the one-step efficient estimator $\hat{r}_\text{eff-1S}$, which inherits its consistency from the initial estimator $\hat{\theta}_\text{init}$ that could be computed using all available moments.

\item It is not restrictive to consider moment matching estimators of the form \eqref{eqn:estimator}. Consider instead any estimator $\hat{\vartheta} \equiv \hat{f}(\hat{\mu})$ of $\theta_0$, where $\hat{f} \colon \mathbb{R}^p \to \mathbb{R}^k$ is a possibly data-dependent function with enough regularity to satisfy the asymptotically linear expansion
\[\hat{\vartheta}-\theta_0 = H(\hat{\mu}-\mu_0) + o_p(n^{-1/2}),\]
for some $k \times p$ matrix $H$. If we restrict attention to asymptotically regular estimators of $\theta_0$ (i.e., estimators that remain asymptotically unbiased under locally drifting parameters), we need $HG = I_k$. Among all estimators $\hat{\vartheta}$ satisfying these requirements, the smallest possible worst-case asymptotic standard deviation of $r(\hat{\vartheta})$ is achieved by the estimator whose asymptotic linearization matrix $H$ solves
\[\min_{H \colon HG=I_k}\;\max_{\tilde{V} \in \mathcal{S}(\diag(V))} (\lambda' H'\tilde{V}H\lambda)^{1/2}.\]
\cref{thm:onestep} in \cref{sec:appendix_proofs} shows that the solution to this problem is precisely the value of the median regression \eqref{eqn:lad}. In other words, the minimum distance estimator $\hat{\theta}$ with (limited-information) efficient weight matrix delivers the smallest possible worst-case standard errors in a large class of estimators.

\end{enumerate}

\subsection{General knowledge of the covariance matrix} \label{sec:ext}
We now extend our results to the general case where any collection of elements of the asymptotic covariance matrix $V$ of the moments $\hat{\mu}$ is known (or consistently estimable), while the remaining elements are unrestricted. For example, if a pair of elements of $\hat{\mu}$ are known to be independent, the corresponding off-diagonal elements of $V$ must equal zero.

Letting $\tilde{\mathcal{S}}$ denote the given constraint set for $V$, we can compute the worst-case asymptotic standard deviation of $r(\hat{\theta})$ as
\begin{equation} \label{eqn:se_general}
	\sqrt{\max_{\tilde{V} \in \tilde{\mathcal{S}}}\, x'\tilde{V}x},
\end{equation}
where $x$ was defined in \cref{sec:se}.\footnote{Similarly, we could compute the \emph{best-case} standard deviation by minimizing this objective function.} In the case of interest to us, $\tilde{\mathcal{S}}$ is defined by equality restrictions on a subset of the elements of $V$, in addition to the restriction that $V$ is symmetric positive semidefinite.\footnote{In fact, the argument also extends to imposing \emph{inequality} constraints on elements of $V$.} Then the maximization problem \eqref{eqn:se_general} is a so-called semidefinite programming problem, a special case of convex programming. In \cref{sec:appendix_ext}, we show that the efficient weight matrix can likewise be computed through a pair of nested convex optimizations. We also show that analytical simplifications obtain in the special case where we know the \emph{block} diagonal of $V$; such a structure may occur if consecutive elements of $\hat{\mu}$ are obtained from the same data set, allowing estimation of covariances among these elements. Our online software suite automatically implement all these computations.

\section{Testing} \label{sec:test}
In this section we develop a joint test of multiple parameter restrictions as well as a test of over-identifying restrictions.

\subsection{Joint testing of parameter restrictions}
\label{sec:test_joint}

We propose a test of the joint null hypothesis $H_0 \colon r(\theta_0)=0_{m \times 1}$ against the two-sided alternative. In this section, the continuously differentiable function $r \colon \Theta \to \mathbb{R}^m$ is allowed to be multi-valued. Tests of a single parameter restriction ($m=1$) can be carried out using the confidence interval described in \cref{sec:se}. For the case $m>1$, we propose the following testing procedure. Let $\alpha \in (0,1)$ denote the significance level.
\begin{enumerate}[i)]
\item Compute the Wald-type test statistic
\[\hat{\mathscr{T}} \equiv r(\hat{\theta})'\hat{S}r(\hat{\theta}),\]
where $\hat{S}$ is a user-specified symmetric positive definite $m \times m$ matrix, to be discussed below.
\item Compute the critical value
\begin{equation} \label{eqn:max_semidef}
\text{cv}_n \equiv \max_{\tilde{V} \in \mathcal{S}(\diag(V))} \frac{1}{n} \tr\left(\tilde{V} WG(G'WG)^{-1}\lambda S\lambda'(G'WG)^{-1}G'W\right) \times \left(\Phi^{-1}(1-\alpha/2) \right)^2.
\end{equation}
In practice, we substitute the estimates $\frac{1}{n} \diag(V) \approx \diag(\hat{\sigma}_1^2,\dots,\hat{\sigma}_p^2)$, $W \approx \hat{W}$, $S \approx \hat{S}$, $G \approx \hat{G} \equiv \frac{\partial h(\hat{\theta})}{\partial \theta'}$, and $\lambda \approx \hat{\lambda} \equiv \frac{\partial r(\hat{\theta})'}{\partial \theta}$.

\item Reject $H_0 \colon r(\theta_0)=0_{m \times 1}$ if $\hat{\mathscr{T}} > \text{cv}_n$.
\end{enumerate}
The maximization \eqref{eqn:max_semidef} is a semidefinite programming problem, which is easy to compute numerically, as discussed in \cref{sec:ext}. \cref{sec:appendix_ext} extends the procedure to settings where additional knowledge about $V$ is available other than the diagonal.

The following proposition shows that, for conventional significance levels $\alpha$, the asymptotic size of this test does not exceed $\alpha$, regardless of the true correlation structure of the moments.
This holds for any valid choice of weight matrix $\hat{W}$, including---but not limited to---the limited-information efficient weight matrix derived in \cref{sec:weight}.

\begin{prop} \label{thm:test}
Impose \cref{asn:regularity}, except that we redefine $\lambda \equiv \partial r(\theta_0)'/\partial \theta$ and  require this matrix to have full column rank $m$. Assume also that $\hat{S} \stackrel{p}{\to} S$, $S$ is symmetric positive definite, $n\,\mathrm{cv}_n>0$, and $\alpha \leq 0.215$. Then, if $r(\theta_0)=0_{m \times 1}$,
\[\limsup_{n \to \infty} P(\hat{\mathscr{T}} > \mathrm{cv}_n) \leq \alpha.\]
\end{prop}
\begin{proof}
Please see \cref{sec:appendix_proofs}.
\end{proof}

\paragraph{Remarks.}
\begin{enumerate}
\item We do not have formal results on how to choose the weight matrix $\hat{S}$ in the test statistic. A pragmatic \emph{ad hoc} choice is to set $\hat{S} = (\lambda'(G'WG)^{-1}G'W\bar{V}WG(G'WG)^{-1}\lambda)^{-1}$ (with consistent estimates plugged in), where $\bar{V}$ is a particular value for the asymptotic variance $V$ of $\hat{\mu}$ at which we wish to direct power. Then $\hat{\mathscr{T}}$ is the usual full-information Wald statistic given $V=\bar{V}$, though the critical value for the test differs from the conventional one to ensure size control regardless of the actual unknown correlation structure. For example, if we want to direct power toward the case with independent moments, we can choose $\bar{V} = \diag(n\hat{\sigma}_1^2,\dots,n\hat{\sigma}_p^2)$.

\item The test procedure in \cref{thm:test} is generally conservative from a minimax perspective, i.e., the size may be strictly smaller than $\alpha$ for all covariance matrices $V$ of the moments. The reason is that the proof of \cref{thm:test} relies on a tail probability bound for quadratic forms of Gaussian random vectors \citep{Szekely2003}. This bound is attained when $V$ has rank 1, but the positive semidefinite maximum \eqref{eqn:max_semidef} need not be attained at a rank-1 matrix, to our knowledge. It is an interesting topic for future research to devise a test that has a formal minimax optimality property given the limited knowledge of $V$.

\item The test procedure is consistent against any fixed alternative with $r(\theta_0) \neq 0_{m \times 1}$ under the conditions of \cref{thm:test}. This follows from the standard argument that $n\hat{\mathscr{T}}$ diverges to infinity with probability 1 in this case, while $\text{cv}_n = O(n^{-1})$ since the largest eigenvalue of any matrix $\tilde{V} \in \mathcal{S}(\diag(V))$ is bounded above by $\sum_{j=1}^p V_{jj}$.

\item An alternative way to test multiple hypotheses ($m>1$) is to separately test each null hypothesis $r_j(\theta_0)=0$, $j=1,\dots,m$, using the univariate confidence interval for $r_j(\theta_0)$ in \cref{sec:se}, but with a Bonferroni correction applied to the significance level.\footnote{We thank a referee for this suggestion. Yet another approach was suggested to us by Bo Honor\'{e}: Reject whenever $\inf_V nr(\hat{\theta})'(\lambda'(G'WG)^{-1}G'WVWG(G'WG)^{-1}\lambda)^{-1}r(\hat{\theta})$ exceeds the $1-\alpha$ quantile of a chi-squared distribution with $m$ degrees of freedom. That is, we search over $V$ for the smallest conventional Wald test statistic. Unfortunately, this optimization problem appears to be numerically challenging unless $p$ is small.} That is, the test rejects the joint hypothesis when any of the $m$ level-$(1-\alpha/m)$ confidence intervals excludes 0. This Bonferroni test neither dominates, nor is it dominated by, the test in \cref{thm:test} in terms of local power in general. The latter test has the advantage that it is finite-sample invariant to nonsingular linear transformations of the joint null hypothesis (if $\hat{S}$ is chosen as in the first remark above).
\end{enumerate}

\subsection{Over-identification testing}
\label{sec:test_overid}
The fit of the calibrated model can be evaluated using over-identification tests when we have more moments $p$ than parameters $k$. In this subsection we allow for potential model misspecification by dropping the assumption in \cref{sec:model} that there exists $\theta_0 \in \mathbb{R}^k$ such that $h(\theta_0)=\mu_0$. Let an arbitrary weight matrix $\hat{W} \stackrel{p}{\to} W$ be given, such as the limited-information efficient weight matrix derived in \cref{sec:weight}. Define the pseudo-true parameter $\tilde{\theta}_0 \equiv \argmin_{\theta \in \mathbb{R}^k} (\mu_0-h(\theta))'W (\mu_0-h(\theta))$, assuming the minimizer is unique. We continue to impose all the assumptions in \cref{sec:model}, with $\tilde{\theta}_0$ substituting for $\theta_0$.

Suppose we want to know whether the model provides a good fit for a particular moment. Let $j^* \in \lbrace 1,\dots,p\rbrace$ be the index of the moment of interest. We seek a confidence interval for the model misspecification measure $\mu_{0,j^*}-h_{j^*}(\tilde{\theta}_0)$, i.e., the $j^*$-th element of $\mu_0-h(\tilde{\theta}_0)$. It is standard to show that, under \cref{asn:regularity},
\begin{equation} \label{eqn:overid_expansion}
\hat{\mu}-h(\hat{\theta})-(\mu_0-h(\tilde{\theta}_0)) = (I_p - G(G'WG)^{-1}G'W)(\hat{\mu}-\mu_0) + o_p(n^{-1/2}).
\end{equation}
Let $\bar{x}$ be the $j^*$-th column of the matrix $I_p - \hat{W}\hat{G}(\hat{G}'\hat{W}\hat{G})^{-1}\hat{G}'$. Then
\[\left[\hat{\mu}_{j^*}-h_{j^*}(\hat{\theta}) - \Phi^{-1}(1-\alpha/2)\widehat{\text{se}}(\bar{x}), \hat{\mu}_{j^*}-h_{j^*}(\hat{\theta}) + \Phi^{-1}(1-\alpha/2)\widehat{\text{se}}(\bar{x}) \right]\]
is a confidence interval for the difference $\mu_{0,j^*}-h_{j^*}(\tilde{\theta}_0)$, with worst-case asymptotic coverage probability $1-\alpha$. Note that it can happen that $\widehat{\text{se}}(\bar{x})=0$, in which case it is not possible to test the over-identifying restriction corresponding to the $j^*$-th moment.\footnote{Similarly, following the ideas in \cref{sec:test_joint}, a \emph{joint} test of the over-identifying restrictions can be constructed by computing the test statistic $\hat{\mathscr{T}}_\text{overid} \equiv (\hat{\mu}-h(\hat{\theta}))'\hat{S}(\hat{\mu}-h(\hat{\theta}))$ for some $p \times p$ symmetric positive definite matrix $\hat{S}$. A natural \emph{ad hoc} choice is $\hat{S}=\hat{W}$, in which case the test statistic equals the minimized minimum distance objective function. We reject correct specification of the model at significance level $\alpha \leq 0.215$ if the test statistic exceeds the critical value $\text{cv}_{n,\text{overid}} \equiv \max_{\tilde{V} \in \mathcal{S}(\diag(V))} \frac{1}{n} \tr\left(\tilde{V} (I_p - WG(G'WG)^{-1}G')S(I_p - G(G'WG)^{-1}G'W)\right) \times \left(\Phi^{-1}(1-\alpha/2) \right)^2$, where we plug in sample analogues for all the unknown quantities. \label{fn:test_overid_joint}}

One common use of over-identification testing is to evaluate the estimated model's fit on ``non-targeted moments''. This corresponds to the special case where the weight matrix $\hat{W}$ zeroes out the corresponding rows and columns of the non-targeted moments, so that the point estimate $\hat{\theta}$ ignores these moments.\footnote{Note that in this case $p$ continues to denote the total number of moments (``targeted'' plus ``non-targeted''), and in particular $\hat{G}$ should contain derivatives of both kinds of moments.} An application of this idea is to fix the structural parameters $\theta$ \alert{to extraneous values---perhaps obtained from previous studies---}and then evaluate whether the model's predictions at these parameters match additional moments (this can be viewed as a statistical version of the model validation process advocated by \citealp{Kydland1996}). This is achieved by choosing a moment function of the form $h(\theta) = (\theta',\bar{h}(\theta))'$, with empirical moments $\hat{\mu}=(\hat{\tilde{\theta}}',\hat{\bar{\mu}})'$. Here $\hat{\tilde{\theta}}$ are the \alert{extraneous parameter values}, while $\bar{h}(\theta)$ is the ``non-targeted'' model prediction of the empirical moment $\hat{\bar{\mu}}$.

\section{Empirical applications}
\label{sec:appl}
We illustrate our methods with two empirical examples. First we fit a model of menu cost price setting to scanner data. Then we fit a heterogeneous agent New Keynesian model to impulse responses estimated from a combination of micro and macro data.

\subsection{Menu cost price setting in multi-product firms}
\label{sec:appl_price}
Our first application estimates the \citet{Alvarez2014} model of menu cost price-setting in multiproduct firms. We fit the model to moments of price changes from supermarket scanner data. This is a small-scale application with $k=3$ parameters and $p=4$ moments.

Though we in fact have access to the underlying data set, we emulate a hypothetical situation where the model is matched to moments that were reported in another paper. We can therefore compare full-information inference, which uses the underlying data, with limited-information inference, which uses only the moments and their marginal standard errors. We find that limited-information inference remains informative about the structural parameters. Moreover, a simulation study calibrated to this application confirms the utility of our procedures in finite samples.

\paragraph{Model.}
We give a brief overview of the structural model here and refer the reader to \citet{Alvarez2014} for details. A firm sets prices on $N$ products. The \emph{desired} log prices for the products evolve in continuous time as $N$ independent Brownian motions (without drift); however, the actual prices are fixed until the firm pays a fixed menu cost, at which point it may reset all $N$ prices simultaneously. The firm's profit depends negatively on the squared log deviation between the current and desired price, integrated over time and averaged across the $N$ products. The model has $k=3$ parameters: the number $N$ of products, the volatility of the desired prices, and the scaled menu cost (relative to the curvature of the profit function).\footnote{In the notation of \citet{Alvarez2014}, these parameters are $n$, $\sigma$, and $\sqrt{\psi/B}$, respectively.} \citet{Alvarez2014} derive closed-form expressions for the frequency of price changes and the moments of the size of price changes. These are the moments we will match in the data.

\paragraph{Data.}
We empirically estimate the frequency and moments of price changes in scanner data from the supermarket chain Dominick's.\footnote{The data is provided by the James M. Kilts Center, University of Chicago Booth School of Business.} As described in detail in \cref{sec:appendix_appl_price}, we clean the data following \citet{Alvarez2016}, and in particular we focus on data from a single store. Unlike those authors, we exclusively use data on beer products, which arguably increases the interpretability of the results and makes the sample size more relevant for our subsequent simulation study. The final data set contains weekly prices on 499 beer products (Universal Product Codes, henceforth UPCs), observed for an average of 76 weeks per UPC. The total sample size is $n=\text{37,916}$. When computing standard errors, we treat the price changes as i.i.d.\ across UPCs and time.

The $p=4$ reduced-form moments that we match to the structural model are the average number of price changes per week as well as the empirical first, second, and fourth moments of the absolute log price changes (conditional on a nonzero change).\footnote{Before computing moments, we subtract off the overall average log price change (conditional on a nonzero change), since the \citet{Alvarez2014} model abstracts from inflation.} We estimate the full-information covariance matrix of these moments using the usual nonparametric estimate (which depends on sample moments of price changes up to order 8). When applying our limited-information procedures, we use only the diagonal of this covariance matrix.

\paragraph{Results.}
We consider both just-identified and efficient specifications. We treat the number $N$ of products as a parameter to be estimated, since there may not be a perfect correspondence between a UPC and the structural model's notion of a ``product''. The just-identified specification uses a weight matrix that attaches zero weight to the first moment of absolute price changes (i.e., the average), so that the three parameters are estimated from three moments (this estimator is available in closed form). We can then check whether the model provides a good fit for the ``non-targeted'' moment by carrying out the over-identification test proposed in \cref{sec:test_overid}. The efficient specification exploits all four empirical moments, using either the conventional full-information one-step estimator or our limited-information one-step estimator in \cref{sec:weight}. In addition to the full-information and limited-information procedures, we report results for a procedure that (erroneously) assumes that the four empirical moments are mutually independent.

\begin{table}[t]
\centering
\textsc{Price setting application: Parameter estimates} \\[0.5\baselineskip]
\renewcommand{\arraystretch}{1.2}
\begin{tabular}{|l||ccc|c|ccc|}
\hline
 & \multicolumn{4}{c|}{Just-identified specification} & \multicolumn{3}{c|}{Efficient specification} \\
 \cline{2-8}
 & {\#} prod. & Vol. & Menu cost & Over-ID & {\#} prod. & Vol. & Menu cost \\
 \hline\hline
Full-info & 3.012	& 0.090	& 0.291	& 0.002	& 3.255	& 0.089	& 0.305 \\
& (0.046) &	(0.001) & (0.003) & (0.000) & (0.051) & (0.001) & (0.003) \\
Independ. & 3.012	& 0.090	& 0.291	& 0.002	& 2.829	& 0.090	& 0.280 \\
& (0.167) & (0.001) & (0.010) & (0.001) & (0.091) & (0.000) & (0.006) \\
Worst case & 3.012	& 0.090	& 0.291	& 0.002	& 2.786	& 0.090	& 0.278 \\
& (0.235) & (0.001) & (0.016) & (0.002) & (0.148) & (0.001) & (0.011) \\
 \hline
\end{tabular}
\caption{Estimates for the just-identified specification (uses only three moments for estimation) and the efficient specification (exploits all four moments for estimation). The rows correspond to full-information inference (exploits knowledge of $\hat{V}$), inference under independence (erroneously assumes that $\hat{V}$ is diagonal), and worst-case inference (exploits only diagonal of $\hat{V}$ without assuming off-diagonal elements are zero). Parameters: number of products (``{\#} prod.''), volatility of desired log price (``Vol.''), scaled menu cost (``Menu cost''). Column ``Over-ID'' displays the error in fitting the non-targeted mean absolute price change moment, given the just-identified parameter estimates. Standard errors in parentheses.} \label{tab:appl_price_estim}
\end{table}

\cref{tab:appl_price_estim} shows that the limited-information standard errors are larger than the full-information ones, but they remain highly informative about the values of the structural parameters. In the efficient over-identified specification, the worst-case standard errors are at most 3.7 times larger than the corresponding full-information values. Importantly, all worst-case standard errors are arguably small relative to the economic magnitudes of the parameter estimates. Hence, taking a worst-case perspective still allows for informative inference. In this particular application, the standard errors that assume independence are mostly intermediate between the full-information and limited-information values.

Though limited-information inference is informative about the structural parameters themselves, there is a price to pay for the over-identification test in this application. In particular, \cref{tab:appl_price_estim} shows that the limited-information test does not reject the validity of the non-targeted moment restriction, whereas the full-information test does reject (however, the economic magnitude of the moment violation is small, as the empirical moment equals 0.145 but the error in fitting the moment is only 0.002).\footnote{Consistent with the over-identification test, some of the full-information efficient parameter estimates are significantly different from the corresponding independence-based or limited-information estimates in \cref{tab:appl_price_estim}. If the model were correctly specified, there should be no statistically significant differences between estimates obtained with different weight matrices. However, the discrepancies between the different estimates are arguably small in economic terms.} This illustrates the principle that full-information inference is usually preferable if it is practically feasible. Of course, if we did not have access to the underlying supermarket scanner data, there would be no alternative to the limited-information analysis.

\paragraph{Simulation study.}
In \cref{sec:appendix_appl_price} we show that our inference procedures perform well in simulations. We simulate data from the \citet{Alvarez2014} model conditional on the estimated structural parameters. While our limited-information tests and confidence intervals have approximately correct size/coverage given the empirical sample size $n$ (as do the full-information procedures), the procedures that erroneously assume independence between the reduced-form moments can over-reject/under-cover.

\subsection{Heterogeneous agent New Keynesian model}
\label{sec:appl_hank}
Our second application estimates a heterogeneous agent New Keynesian general equilibrium macro model, following \citet{McKay2016} and \citet{Auclert2020}. The matched moments are impulse response functions of macro time series and cross-sectional micro moments with respect to identified productivity and monetary policy shocks, as estimated by \citet{Chang2021} and \citet{MirandaAgrippino2021}. This is a medium-scale application with $k=7$ parameters and $p=23$ moments.

Though less efficient, impulse response matching estimation is more robust to modeling assumptions than full-information likelihood estimation.\footnote{Likelihood procedures for estimation of heterogeneous agent models have been proposed by \citet{Mongey2017}, \citet{Winberry2018}, \citet{Auclert2020}, and \citet{Liu2023} among others.} This is because---in the first-order approximation we consider---the impulse responses with respect to a monetary shock, say, do not depend on the exogenous processes for the other disturbances (e.g., shocks to the household discount rate).\footnote{Impulse responses to a monetary shock are computed by holding fixed all other exogenous shocks. As a result, the linearized impulse responses depend only on parameters of the monetary disturbance process as well as model parameters that govern the endogenous transmission mechanisms.} Thus, our application only requires us to specify and estimate the exogenous processes for productivity and monetary disturbances. We remain agnostic about the number and nature of other shocks that may be driving the economy.

Likewise, our limited-information approach is simpler and less restrictive than other types of procedures that attempt to exploit more information. The only data inputs into our procedure are the impulse response point estimates and confidence intervals reported by \citet{Chang2021} and \citet{MirandaAgrippino2021}. We do not need access to the underlying data used in those papers, as would be required if one were to estimate the joint covariance matrix of all empirical moments via the bootstrap or GMM calculations. Unlike approaches based on bootstrapping or simulating data, we do not need to repeatedly re-run the impulse response estimation routines, and we do not need to fully model the relationship between the macro and micro data used by \citet{Chang2021} (e.g., by specifying all shocks).

\paragraph{Model.}
We employ the one-asset heterogeneous agent New Keynesian model described in \citet[Appendix B.2]{Auclert2020}. Following \citet{McKay2016}, the model features a continuum of heterogeneous households facing uninsurable idiosyncratic earnings risk. The households choose their work hours and amount of savings in a nominal Treasury bond. Monopolistically competitive firms set prices subject to a quadratic adjustment cost, yielding a New Keynesian Phillips curve. Households receive lump sum distributions of government interest revenue and firm profits. The central bank sets the nominal interest rate according to a Taylor rule that depends on inflation. \alert{A detailed description of the model is provided in the Online Supplement.}

\alert{Our model differs from that of \citet{Auclert2020} in only three respects. First, we set the Taylor rule coefficient on output equal to zero, as this coefficient is estimated to be numerically small by \citet[Table F.III]{Auclert2020}. Second, rather than assuming constant total factor productivity (TFP), we allow for an exogenous AR(2) disturbance to the log growth rate of TFP. Third, in order to allow for more flexible dynamics, we generalize the AR(1) process for the monetary disturbance (the residual in the Taylor rule) assumed by \citeauthor{Auclert2020} to an AR(2) process. \citeauthor{Auclert2020} additionally allow for shocks to government spending and markups; our estimation strategy is robust to the presence of such additional shocks, as noted above, though we do not explicitly estimate the effects of these shocks.}

The model is solved through a first-order linearization, using the numerical procedures developed by \citet{Auclert2020}. That is, the moment function $h(\cdot)$ is computed numerically, and its derivatives are approximated via finite differences; no Monte Carlo simulation is involved.

Following \citet{Auclert2020}, we limit ourselves to estimating structural parameters that do not affect the steady state of the model. This allows us to avoid repeatedly recomputing the steady state, though this would be feasible to do with moderate computational effort. The steady state parameters are fixed at the values assumed by \citet[Table B.2]{Auclert2020}. The $k=7$ estimated parameters are: the Taylor rule coefficient on inflation, the slope of the Phillips curve, the three parameters in the AR(2) process for TFP growth, and the two autoregressive coefficients for the monetary disturbance.\footnote{We do not need to estimate the standard deviation of the monetary shock, since this parameter does not affect the normalized impulse responses that we match (see below).}

\paragraph{Data.}
The empirical moments are obtained from two sets of Structural Vector Autoregression estimates of impulse responses to identified shocks.

Impulse responses with respect to TFP shocks are obtained from \citet[Fig.\ 9 and 11, blue lines]{Chang2021}. We use the responses of TFP itself and of GDP (output in the model), as well as the response of a cross-sectional moment estimated using data from the Current Population Survey (CPS): the fraction of people earning less than 2/3 of per capita GDP.\footnote{The factor 2/3 approximately adjusts for the average labor share, see \citet[Sec.\ 6.1]{Chang2021}.} The sophisticated estimation method of \citet{Chang2021} takes into account statistical uncertainty arising from the limited sample sizes in the CPS. By relying directly on their reported results, our analysis inherits this desirable feature.

Impulse responses with respect to monetary shocks are obtained from \citet[Fig.\ 3]{MirandaAgrippino2021}. We use the responses of industrial production (output in the model), the consumer price index (price level in the model), and the 1-year Treasury rate (annualized nominal interest rate in the model). Since our structural model is quarterly but the \citet{MirandaAgrippino2021} data is monthly, we use the end-of-quarter impulse responses.

We focus on four impulse response horizons: the impact horizon, and the 1-, 2-, and 8-quarter horizons. When matching the model to the data, we take into account that the \citet{Chang2021} responses are with respect to a one-standard-deviation shock, while the \citet{MirandaAgrippino2021} responses are normalized so that the Treasury rate increases by 100 basis points on impact.\footnote{\citet{Chang2021} actually consider a 3-standard-deviation shock, but we divide by 3.} Since both papers report Bayesian posterior quantiles, we appeal to the Bernstein-von Mises theorem and define the point estimates to be the reported posterior medians, while the standard errors are those implied by a normal approximation of the reported credible intervals.\footnote{Thus, if the length of the $1-\alpha$ credible interval for $\theta_j$ is $\hat{L}_j$, we set $\hat{\sigma}_j=\hat{L}_j/(2\Phi^{-1}(1-\alpha/2))$.} In total, we have $p=23$ empirical moments, as we discard the impact response of the bond rate, which is normalized to 1.

\paragraph{Results.}
The top row of \cref{tab:appl_hank_estim} shows the parameter estimates
obtained by using a diagonal weight matrix with
$W_{jj}=1/\hat{\sigma}_j^2$.\footnote{We run gradient-based numerical
optimizations from 100 different starting values and report the overall
optimum.
One vector of starting values equals the posterior mode estimates $\hat{\theta}_i^\text{ABRS}$ from \citet[Table F.III]{Auclert2020} \alert{where possible (including setting the AR2 coefficient for the monetary disturbance to zero as assumed therein), supplemented by parameters implying that TFP growth is white noise with standard deviation
equal to the contemporaneous response of TFP to a TFP
shock reported by \citet{Chang2021}.}
The 99
other starting values are simulated uniformly at random from intervals
of the form $[0,2\hat{\theta}_i^\text{ABRS}]$, or $[-0.99,0.99]$ in the
case of AR1 coefficients.} The Taylor rule coefficient on inflation is
estimated to slightly exceed 1. The slope of the Phillips curve is
positive but statistically insignificant at conventional significance
levels. The TFP growth process is estimated to be close to white noise
(i.e., the \emph{level} of TFP is close to a random walk, as commonly assumed
in the literature), while the monetary disturbance process has a
half-life of about 2 quarters.

\begin{table}[t]
\centering
\textsc{Heterogeneous agent application: Parameter estimates} \\[0.5\baselineskip]
\renewcommand{\arraystretch}{1.2}
\begin{tabular}{|l||c|c|ccc|cc|}
\hline
 & & & \multicolumn{3}{c|}{TFP} & \multicolumn{2}{c|}{Monetary} \\
 \cline{4-8}
Weight matrix & TR & PC & AR1 & AR2 & Std & AR1 & AR2 \\
 \hline\hline
Diagonal & 1.060 & 0.009 & 0.008 & -0.040 & 0.006 & 0.702 & 0.075 \\
& (0.100) & (0.008) & (0.146) & (0.189) & (0.000) & (0.110) & (0.163) \\
\hline
Efficient & 1.040 & 0.015 & 0.007 & -0.017 & 0.006 & 0.697 & 0.014 \\
& (0.074) & (0.006) & (0.137) & (0.152) & (0.000) & (0.088) & (0.127) \\
 \hline
\end{tabular}
\caption{Structural parameter estimates with diagonal weight matrix (top row) and efficient weighting (bottom row). Parameters: Taylor rule coefficient on inflation (``TR''); slope of Phillips curve (``PC''); first and second autoregressive (``AR1'' and ``AR2'') and standard deviation (``Std'') parameters of TFP and monetary disturbance processes. Worst-case standard errors in parentheses.} \label{tab:appl_hank_estim}
\end{table}

\cref{fig:appl_hank_irf} compares the model-implied and empirical impulse responses, at the parameter estimates discussed in the previous paragraph. We see that the model-implied impulse responses of output to a monetary shock are too small in magnitude relative to the data at the 2- and 8-quarter horizons, while the opposite is true for the responses of the price level with respect a monetary shock. Moreover, the model fails to generate nontrivial dynamics in the fraction of low-wage earners in response to a TFP shock. To test whether these disparities are too large to be explained by statistical noise, we conduct the over-identification test proposed in \cref{sec:test_overid}. The vertical error bars in \cref{fig:appl_hank_irf} show the 90\% confidence intervals for the \emph{differences} between model-implied and empirical moments, centered at the empirical moments for visual convenience. Three of the model-implied impulse responses for the fraction of low-wage earners fall outside their respective intervals, though only marginally so. While this points to misspecification of either the structural model or the reduced-form VAR models used to generate the empirical impulse responses, we note that the \emph{joint} test of the validity of all $p=23$ moments does not reject at the 10\% level.\footnote{The test statistic in \cref{fn:test_overid_joint} equals 24.58 with critical value 56.99.}

\begin{figure}[t]
\centering
\textsc{Heterogeneous agent application: Impulse responses} \\[0.5\baselineskip]
\includegraphics[width=0.7\linewidth]{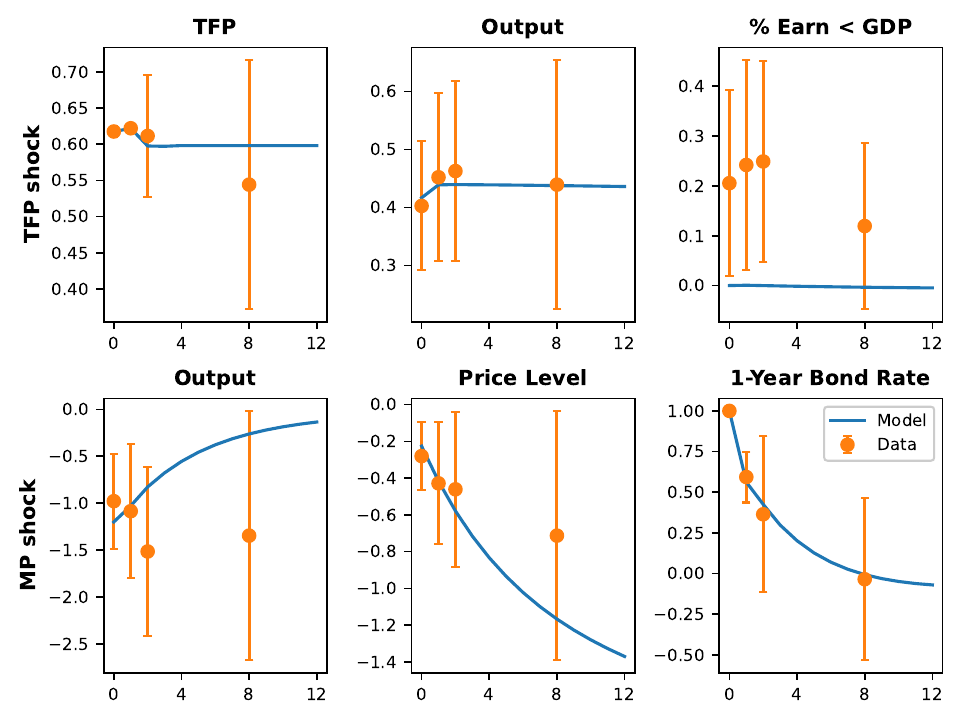}
\caption{Model-implied impulse responses (thin curves) and corresponding empirical estimates (circles), with respect to a one-standard-deviation TFP shock (top row) or a monetary policy shock that raises the bond rate by 1 percentage point on impact (bottom row). Figure titles: response variables (``\% Earn < GDP'': fraction of people earning less than 2/3 of per capita GDP). Vertical axis units: percentage points. Horizontal axis units: quarters. Vertical error bars: (shifted) 90\% confidence intervals for the \emph{differences} between the empirical and model-implied moments (not confidence intervals for the empirical moments themselves).}
\label{fig:appl_hank_irf}
\end{figure}

The efficient (one-step) parameter estimates in the bottom row of \cref{tab:appl_hank_estim} demonstrate the benefit of optimally weighting the moments as described in \cref{sec:weight}. In particular, the t-statistic for the slope of the New Keynesian Phillips curve increases to 2.32, from 1.10 previously. \alert{Thus, our limited-information approach obtains a statistically significant slope parameter, as in the full-information estimation exercise of \citet{Auclert2020} but without requiring specification of all the shocks driving the economy.} More generally, the efficient standard errors in the bottom row of \cref{tab:appl_hank_estim} are 6--26\% smaller than the non-efficient ones in the top row.\footnote{\cref{tab:appl_hank_load} in \cref{sec:appendix_appl_hank} shows which moments are optimally selected to estimate the various parameters.}

\section{Conclusion}
\label{sec:concl}
We computed simple, sharp, and informative upper bounds on the standard errors of structural parameter estimates when the correlation structure of the matched empirical moments is not fully known. In addition, we proposed an efficient moment weighting procedure in the over-identified case, as well as valid tests of parameter restrictions and over-identifying restrictions. The required inputs are minimal: Other than being able to evaluate the mapping from structural parameters to model-implied moments (at least numerically), we just need the empirical moment estimates and their individual standard errors. Our procedures are computationally tractable even in settings with many moments and/or parameters. A code suite is available online.

We believe our limited-information approach is useful for applied researchers who match their models to moments obtained from several different data sources, estimation methods, or previous papers. Our methods obviate the need to estimate the correlation structure across the various moments, which is sometimes difficult or impossible. Even when the moment correlation structure is in principle estimable, our methods may be helpful, since marginal standard errors for individual moments are typically much easier to obtain from standard econometric software than it is to figure out the joint distribution of all moments, as illustrated in our empirical applications. Moreover, the limited-information procedures can be used to gauge whether it is worthwhile to expend the additional effort required for full-information analysis.

\clearpage
\appendix
\section{Appendix} \label{sec:appendix}
\subsection{Technical lemmas and proofs}
\label{sec:appendix_proofs}
Here we state and prove a technical lemma referred to in \cref{sec:weight}, and we provide the proofs of \cref{thm:coverage,thm:test}.

\begin{lem} \label{thm:onestep}
Assume $p,k \in \mathbb{N}$ and $p>k$. Let $\lambda \in \mathbb{R}^k$, and let $G \in \mathbb{R}^{p \times k}$ have full column rank. Let $G^\perp$ denote any $p \times (p-k)$ matrix with full column rank such that $G'G^\perp = 0_{k \times (p-k)}$. Let $\mathcal{S}_p$ denote the set of $p \times p$ symmetric positive semidefinite matrices $W$ such that $G'WG$ is nonsingular. Then
\[\left\lbrace WG(G'WG)^{-1}\lambda \colon W \in \mathcal{S}_p \right\rbrace = \left\lbrace x \colon x \in \mathbb{R}^p,\; G'x=\lambda \right\rbrace = \left\lbrace G(G'G)^{-1}\lambda + G^\perp z \colon z \in \mathbb{R}^{p-k} \right\rbrace.\]
\end{lem}
\begin{proof}
We first show that
\begin{equation} \label{eqn:contain1}
\left\lbrace G(G'G)^{-1}\lambda + G^\perp z \colon z \in \mathbb{R}^{p-k} \right\rbrace \subset \left\lbrace WG(G'WG)^{-1}\lambda \colon W \in \mathcal{S}_p \right\rbrace.
\end{equation}
Pick any $z \in \mathbb{R}^{p-k}$, and define
\[W(z) \equiv (G, G^\perp) \left( \begin{array}{cc}
I_k & \tilde{\lambda}z' \\
z\tilde{\lambda}' & \delta I_{p-k}
\end{array} \right) \left(\begin{array}{c}
G' \\
G^{\perp\prime}
\end{array} \right),\quad \tilde{\lambda} \equiv \frac{1}{\lambda'(G'G)^{-1}\lambda}\lambda,\]
where $\delta>0$ is arbitrary but chosen large enough so that $W(z)$ is positive semidefinite. Then
\[W(z)G = (G, G^\perp) \left( \begin{array}{cc}
I_k & \tilde{\lambda}z' \\
z\tilde{\lambda}' & \delta I_{p-k}
\end{array} \right)\left(\begin{array}{c}
G'G \\
0_{(p-k) \times k}
\end{array} \right) = (G + G^\perp z\tilde{\lambda}
') G'G.\]
Hence,
\[G'W(z)G = (G'G)^2,\]
implying
\begin{align*}
W(z)G(G'W(z)G)^{-1}\lambda &= ( G + G^\perp z\tilde{\lambda}') (G'G)^{-1}\lambda = G(G'G)^{-1}\lambda + G^\perp z,
\end{align*}
and thus the statement \eqref{eqn:contain1} holds.

Now pick any $W \in \mathcal{S}_p$. Then $x = WG(G'WG)^{-1}\lambda$ satisfies $G'x=\lambda$. This shows that
\begin{equation} \label{eqn:contain2}
\left\lbrace WG(G'WG)^{-1}\lambda \colon W \in \mathcal{S}_p \right\rbrace \subset \left\lbrace x \colon x \in \mathbb{R}^p,\; G'x=\lambda \right\rbrace.
\end{equation}
Finally, choose any $x \in \mathbb{R}^p$ satisfying $G'x=\lambda$. Since the columns of $G$ and $G^\perp$ are (jointly) linearly independent, there exist $y \in \mathbb{R}^k$ and $z \in \mathbb{R}^{p-k}$ such that $x = Gy + G^\perp z$. Note that $\lambda = G'x = G'Gy$, so necessarily $y=(G'G)^{-1}\lambda$. We have thus shown that
\begin{equation} \label{eqn:contain3}
\left\lbrace x \colon x \in \mathbb{R}^p,\; G'x=\lambda \right\rbrace \subset \left\lbrace G(G'G)^{-1}\lambda + G^\perp z \colon z \in \mathbb{R}^{p-k} \right\rbrace.
\end{equation}
The set inclusions \eqref{eqn:contain1}--\eqref{eqn:contain3} together imply the statement of the lemma.
\end{proof}

\paragraph{Proof of \cref{thm:coverage}.}
It is standard to show that \eqref{eqn:delta_method} holds under \cref{asn:regularity} \citep{Newey1994}. Moreover,
\[\sqrt{n}\, \widehat{\text{se}}(\hat{x}) \stackrel{p}{\to} \sum_{j=1}^p V_{jj}^{1/2}|x_j| = \max_{\tilde{V} \in \mathcal{S}(\diag(V))} \sqrt{x'\tilde{V}x},\]
where the equality uses \cref{thm:var_bound} in \cref{sec:se}, $\diag(V)$ is the vector of diagonal elements of $V$, and we recall the notation $x = WG(G'WG)^{-1}\lambda$. Hence, letting $Z$ denote a standard normal random variable, we have
\begin{align*}
\pushQED{\qed}
\lim_{n \to \infty}P(r(\theta_0) \in \widehat{\text{CI}}) &= \lim_{n \to \infty}P\left(\sqrt{n}|r(\hat{\theta})-r(\theta_0)| \leq \Phi^{-1}(1-\alpha/2) \times \sqrt{n}\,\widehat{\text{se}}(\hat{x})\right) \\
&= P\left(\sqrt{x'Vx}\times |Z| \leq \Phi^{-1}(1-\alpha/2) \times \max_{\tilde{V} \in \mathcal{S}(\diag(V))} \sqrt{x'\tilde{V}x}\right) \\
&\geq P\left(|Z| \leq \Phi^{-1}(1-\alpha/2)\right) \\
&= 1-\alpha. \qedhere \popQED
\end{align*}

\paragraph{Proof of \cref{thm:test}.}
Under the null hypothesis,
\[\sqrt{n}r(\hat{\theta}) \stackrel{d}{\to} \lambda'(G'WG)^{-1}G'WV^{1/2}Z,\]
where $V^{1/2}V^{1/2\prime}=V$, and $Z=(Z_1,\dots,Z_p)' \sim N(0_{p \times 1},I_p)$. The asymptotic null distribution of the test statistic $\hat{\mathscr{T}}$ is therefore a Gaussian quadratic form:
\[n\hat{\mathscr{T}} \stackrel{d}{\to} Z'QZ,\quad Q \equiv V^{1/2\prime}WG(G'WG)^{-1}\lambda S\lambda'(G'WG)^{-1}G'WV^{1/2}.\]
\citet{Szekely2003} prove that
\begin{equation} \label{eqn:szekely_bakirov}
P(Z'QZ \leq \tr(Q) \times \tau) \geq P(Z_1^2 \leq \tau)
\end{equation}
for \emph{any} $p \times p$ symmetric positive semidefinite (non-null) matrix $Q$ and any $\tau > 1.5365$. Since $(\Phi^{-1}(1-\alpha/2))^2 > 1.5365$ for $\alpha \leq 0.215$, it follows that, under the null,
\begin{align*}
\pushQED{\qed}
P(\hat{\mathscr{T}} \leq \text{cv}_n) &\geq P\left(n\hat{\mathscr{T}} \leq \tr(Q) \times (\Phi^{-1}(1-\alpha/2))^2\right) \\
&\to P\left(Z'QZ \leq \tr(Q) \times (\Phi^{-1}(1-\alpha/2))^2\right) \\
&\geq P\left(Z_1^2 \leq (\Phi^{-1}(1-\alpha/2))^2 \right) \\
&= 1-\alpha. \qedhere \popQED
\end{align*}

\subsection{Simulated moments}
\label{sec:appendix_sim}
We here argue that \cref{thm:coverage} can be extended to a setting where the model-implied moment function $h(\cdot)$ is not available analytically and we therefore resort to stochastic simulation to approximate it. The main requirement is that the number of simulation draws $M$ is sufficiently large relative to the empirical sample size $n$.

Let $\hat{h}_M(\cdot)$ denote the approximation of $h(\cdot)$ computed from $M$ simulation draws. For example, if we are matching simple moments of the form $h(\theta) = E[\tilde{h}(\theta,\zeta)]$, where $\tilde{h}(\cdot,\cdot)$ is a deterministic function and $\zeta$ is a random vector (e.g., the shocks in the structural model), then we can set $\hat{h}_M(\theta)=\frac{1}{M}\sum_{s=1}^M \tilde{h}(\theta,\zeta_s)$, where $\zeta_1,\dots,\zeta_M$ are draws from the distribution of $\zeta$. The minimum distance estimator $\hat{\theta}$ is defined as in \eqref{eqn:estimator}, but with $\hat{h}_M(\cdot)$ in place of $h(\cdot)$. The worst-case standard errors are computed as in \eqref{eqn:wcse}, except that we leave the estimator $\hat{G}=\hat{G}_M$ of $G=\frac{\partial h(\theta_0)}{\partial \theta'}$ unspecified for now (we return to this choice below).

\begin{prop} \label{thm:coverage_sim}
Impose \cref{asn:regularity} and $\max_j |x_j|V_{jj}>0$. Assume further that as $n \to \infty$, we have $M \to \infty$, $n/M \to 0$, $\hat{G}_M \stackrel{p}{\to} G$, and $\sup_{\theta \colon \|\theta-\theta_0\|\leq \delta_M}\sqrt{M}\|\hat{h}_M(\theta)-h(\theta)\| = O_p(1)$ for any sequence satisfying $\delta_M \to 0$. Then the conclusions of \cref{thm:coverage} apply.
\end{prop}
\begin{proof}
Upon inspection of the proof of \cref{thm:coverage}, it suffices to show that \eqref{eqn:delta_method} holds. We appeal to Theorem 7.2 in \citet{Newey1994}, where, in their notation, $\hat{g}_n(\theta)=\hat{h}_M(\theta)-\hat{\mu}$ and $g_0(\theta)=h(\theta)-h(\theta_0)$. The first three conditions of this theorem hold by \cref{asn:regularity}(\ref{itm:asn_regularity_3})--(\ref{itm:asn_regularity_6}). The fourth condition holds by \cref{asn:regularity}(\ref{itm:asn_regularity_1}) and $\sqrt{n}[\hat{h}_M(\theta_0)-h(\theta_0)] = \sqrt{n/M} \times \sqrt{M}[\hat{h}_M(\theta_0)-h(\theta_0)] = o(1) \times O_p(1)$. Finally, the fifth condition follows from
\begin{align*}
&\sup_{\theta \colon \|\theta-\theta_0\|\leq \delta_M} \sqrt{n}\|\hat{g}_n(\theta)-\hat{g}_n(\theta_0)-g_0(\theta)\| = \sup_{\theta \colon \|\theta-\theta_0\|\leq \delta_M} \sqrt{n}\|\hat{h}_M(\theta)-\hat{h}_M(\theta_0)-[h(\theta)-h(\theta_0)]\| \\
&\leq 2\sqrt{n/M} \times \sup_{\theta \colon \|\theta-\theta_0\|\leq \delta_M} \sqrt{M}\|\hat{h}_M(\theta)-h(\theta)\| = o(1) \times O_p(1). \qedhere
\end{align*}
\end{proof}

\paragraph{Remarks.}
\begin{enumerate}
\item The key assumption is that the number of simulation draws $M$ is large relative to the empirical sample size $n$, in the sense $n/M \to 0$, so that the simulation noise in $\hat{h}_M(\cdot)$ is asymptotically negligible relative to the statistical noise in $\hat{\mu}$. In practice, we recommend that researchers choose $M$ sufficiently large so that the simulation noise in the moment estimates $\hat{h}_M(\hat{\theta})$ is orders of magnitude smaller than the empirical moment standard errors $\hat{\sigma}_j$.
\item The uniform tightness assumption $\sup_{\theta \colon \|\theta-\theta_0\|\leq \delta_M}\sqrt{M}\|\hat{h}_M(\theta)-h(\theta)\| = O_p(1)$ strengthens pointwise $\sqrt{M}$-consistency of $\hat{h}_M(\theta)$ (which usually follows straight-forwardly from a central limit theorem) to uniform $\sqrt{M}$-consistency for all $\theta$ in a neighborhood of $\theta_0$.
\item Under alternative asymptotics where $M$ is proportional to $n$, \alert{the simulation noise in the moments must be taken into account in the worst-case standard error formula}, and the uniform tightness assumption would need to be strengthened to stochastic equicontinuity as in \citet[Theorem 7.2]{Newey1994}. \alert{In practice, our procedure can be applied without further modification if, for each moment, we increase the squared moment standard error $\hat{\sigma}_j^2$ by the variance of the Monte Carlo simulation noise of that moment.}
\item Consider again the special case with $h(\theta) = E[\tilde{h}(\theta,\zeta)]$ and $\hat{h}_M(\theta)=\frac{1}{M}\sum_{s=1}^M \tilde{h}(\theta,\zeta_s)$, where we now assume that the draws $\lbrace \zeta_s\rbrace_{s=1}^M$ are strictly stationary across $s$. Then uniform tightness is implied by $\sum_{\ell=-\infty}^\infty \|\cov(\tilde{h}(\theta,\zeta_s),\tilde{h}(\theta,\zeta_{s-\ell}))\|$ being uniformly bounded for all $\theta$ in a neighborhood of $\theta_0$ (this follows from Chebyshev's inequality). Notice that these conditions allow the simulated moments $\hat{h}_M(\theta)$ to be computed either by averaging across time (in a single simulated economy) or across independent simulated economies.
\item The only assumption required of the estimated Jacobian matrix $\hat{G}_M$ is consistency. \citet[Section 7.3]{Newey1994} and \citet{Hong2015} discuss generic finite-difference derivative estimators based on simulation draws, and they provide assumptions ensuring consistency.
\end{enumerate}

\subsection{Details of extensions}
\label{sec:appendix_ext}
Here we provide further details on the cases discussed in \cref{sec:ext} where additional knowledge about the moment covariance matrix is available.

\paragraph{General constraint set.}
Consider a general constraint set $\tilde{\mathcal{S}}$ defined by linear equality or inequality restrictions, as well as the positive semidefiniteness constraint. In the over-identified case $p>k$, the worst-case efficient weight matrix $W$ can be computed through two nested convex/concave optimization problems:
\begin{equation} \label{eqn:weight_general}
	\min_{W \in \mathcal{S}_p} \max_{\tilde{V} \in \tilde{\mathcal{S}}}\, x(W)'\tilde{V}x(W) = \min_{z \in \mathbb{R}^{p-k}} \max_{\tilde{V} \in \tilde{\mathcal{S}}}\, \lbrace G(G'G)^{-1}\lambda + G^\perp z\rbrace'\tilde{V}\lbrace G(G'G)^{-1}\lambda + G^\perp z\rbrace,
\end{equation}
where $\mathcal{S}_p$, $x(W)$, and $G^\perp$ were defined in \cref{sec:weight}, and the equality follows from \cref{thm:onestep} in \cref{sec:appendix_proofs}. The inner maximization in \eqref{eqn:weight_general} is a concave semidefinite program, as discussed in \cref{sec:ext}. The outer minimization is an unconstrained convex program since the objective function is a pointwise maximum of convex functions in $z$. The nested optimizations in \eqref{eqn:weight_general} may be neither \emph{strictly} convex nor differentiable, but for our purposes it suffices to use convex optimization algorithms that return any arbitrary local minimum $z^*$ (which is necessarily also a global minimum). Once an optimal $z^*$ has been computed, a corresponding optimal weight matrix is given by the matrix $W(z^*)$ defined in the proof of \cref{thm:onestep} in \cref{sec:appendix_proofs}.

To conduct joint hypothesis tests as in \cref{sec:test}, we can simply replace the constraint set in the critical value computation \eqref{eqn:max_semidef} with $\tilde{\mathcal{S}}$.

\paragraph{Special case: Knowledge of the block diagonal.}
Suppose we know the \emph{block} diagonal of $V$, while all other elements are unrestricted. That is, suppose the constraint set $\tilde{\mathcal{S}}$ is given by all symmetric positive semidefinite matrices of the form
\begin{equation} \label{eqn:block_diag}
	V = \left( \begin{array}{ccccc}
		V_{(1)} & ? & ? & \dots & ? \\
		? & V_{(2)} & ? & \dots & ? \\
		\vdots & &  \ddots & & \vdots \\
		\vdots & & & \ddots & \vdots \\
		? & ? & \dots & ? & V_{(J)}
	\end{array} \right),
\end{equation}
where $V_{(j)}$ are known (or consistently estimable) square symmetric matrices (possibly of different dimensions) for $j=1,\dots,J$. Partition the vectors $x$ and $\hat{\mu}$ conformably as $x=(x_{(1)}',\dots,x_{(J)}')'$ and $\hat{\mu} = (\hat{\mu}_{(1)}',\dots,\hat{\mu}_{(J)}')'$. The worst-case asymptotic standard deviation \eqref{eqn:se_general}, for fixed $W$, is then given by
\begin{equation} \label{eqn:se_block}
	\sqrt{\max_{\tilde{V} \in \tilde{\mathcal{S}}}\, x'\tilde{V}x}=\sum_{j=1}^J (x_{(j)}'V_{(j)}x_{(j)})^{1/2}.
\end{equation}
This follows from the same logic as in \cref{thm:var_bound} in \cref{sec:se} once we recognize that the known block diagonal of $V$ implies that the marginal variance of $x_{(j)}'\hat{\mu}_{(j)}$ is known for each $j=1,\dots,J$, but the correlations among these $J$ variables remain unrestricted. Specifically, the maximum \eqref{eqn:se_block} is achieved by $V=\var(\tilde{\mu})$, where the random vector $\tilde{\mu}=(\tilde{\mu}_{(1)}',\dots,\tilde{\mu}_{(J)}')'$ has the following representation. Let $\eta=(\eta_{(1)}',\dots,\eta_{(J)}')'$ have the covariance matrix \eqref{eqn:block_diag}, but with zeros instead of question marks. Let $\bar{\eta}$ be a scalar random variable with variance 1 that is uncorrelated with $\eta$. Then set $\tilde{\mu}_{(j)} \equiv \frac{1}{\sqrt{x_{(j)}'V_{(j)}x_{(j)}}}V_{(j)}x_{(j)}\bar{\eta} + (I - \frac{1}{x_{(j)}'V_{(j)}x_{(j)}} V_{(j)}x_{(j)}x_{(j)}')\eta_{(j)}$, $j=1,\dots,J$.

In the over-identified case, the worst-case efficient weight matrix can be computed by substituting the formula \eqref{eqn:se_block} into the nested optimization \eqref{eqn:weight_general} (with $x=G(G'G)^{-1}\lambda + G^\perp z$).

\subsection{Details of the price setting application and simulation study}
\label{sec:appendix_appl_price}
Here we provide details of the data used for the empirical application in \cref{sec:appl_price}, and we conduct a simulation study calibrated to this application.

\paragraph{Data.}
We use the ``movement'' data set for beer products (file name {\tt wber.csv}) on the Chicago Booth website.\footnote{\url{https://www.chicagobooth.edu/research/kilts/datasets/dominicks}} We follow \citet{Alvarez2016} when cleaning the data. First, we keep only data for store {\#}122. Second, we drop any observations with prices below 20 cents or above 25 dollars (the data was collected between the years 1989 and 1994). Third, we set any absolute price changes below one cent equal to zero. Fourth, we drop the largest 1\% of absolute log price changes.

\begin{table}[t]
\centering
\textsc{Price setting application: Reduced-form moment estimates} \\[0.5\baselineskip]
\renewcommand{\arraystretch}{1.2}
\begin{tabular}{|l||c|c|ccc|}
\hline
Moment & Estimate & Std. error & \multicolumn{3}{c|}{Pairwise correlation} \\
 &  & $\times 1000$ & $E[(\Delta p)^2]$ & $E[(\Delta p)^4]$ & $E[|\Delta p|]$ \\
\hline\hline
Frequency &	0.293	&	2.338	&	0.000	&	0.000	&	0.000	\\
$E[(\Delta p)^2]$ &	0.027	&	0.233	&		&	0.939	&	0.966	\\
$E[(\Delta p)^4]$ &	0.001	&	0.019	&		&		&	0.831	\\
$E[|\Delta p|]$ &	0.145	&	0.754	&	&		&		\\
\hline

\end{tabular}
\caption{Reduced-form moment estimates and their standard errors: average weekly rate of nonzero price changes (Frequency), and moments $E[|\Delta p|^j]$, $j=1,2,4$ of absolute log price changes conditional on a nonzero change. ``Pairwise correlation'' columns: estimated pairwise correlations across the sample moments.} \label{tab:appl_price_moments}
\end{table}

\cref{tab:appl_price_moments} shows the $p=4$ estimated reduced-form moments, their standard errors, and their estimated correlation matrix. The sample kurtosis (fourth moment divided by squared second moment) of log price changes equals 1.80. The zero correlation between the sample frequency of nonzero price changes (a binary outcome) and the sample moments of the price change magnitudes is mechanical. Our limited-information analysis here does not exploit this fact because we want to emulate what an applied researcher might do without thinking hard about the problem. However, the independence could be taken into account using the extensions described in \cref{sec:ext}.

\paragraph{Simulation study.}
We apply the inference methods to data simulated from the \citet{Alvarez2014} model. The simulations treat the just-identified empirical parameter estimates (columns 1--3 in \cref{tab:appl_price_estim}) as the truth, and we use a sample size of $n=\text{37,916}$ as in the real data. The binary price change indicators are drawn i.i.d.\ from a binomial distribution with the model-implied success probability \citep[Proposition 4]{Alvarez2014}. The magnitudes of the price changes are drawn from the model-implied density function \citep[Proposition 6]{Alvarez2014}.\footnote{We simulate from this density by numerically computing the associated quantile function on a fine grid, and then passing random uniform draws through a cubic interpolation of this function.} We use 10,000 Monte Carlo repetitions. The estimation and inference procedures are the same as the ones applied to the actual data (in particular, efficient estimates are computed using the one-step approach).

\begin{table}[p]
\centering
\textsc{Monte Carlo simulation study} \\[0.5\baselineskip]
\renewcommand{\arraystretch}{1.2}
\begin{tabular}{|l|ccc|ccc|}
\hline
 & \multicolumn{3}{c|}{Just-identified specification} & \multicolumn{3}{c|}{Efficient specification} \\
 \hline
 & {\#} prod. & Vol. & Menu cost & {\#} prod. & Vol. & Menu cost \\
\hline\hline
& \multicolumn{6}{c|}{ } \\[-2ex]
& \multicolumn{6}{c|}{\underline{Confidence interval coverage rate}} \\[1ex]
Full-info	&	94.5\%	&	95.0\%	&	94.7\%	&	95.0\%	&	95.2\%	&	95.3\%	\\
Independence	&	100.0\%	&	95.0\%	&	100.0\%	&	100.0\%	&	89.1\%	&	100.0\%	\\
Worst case	&	100.0\%	&	99.4\%	&	100.0\%	&	100.0\%	&	99.4\%	&	100.0\%	\\

& \multicolumn{6}{c|}{ } \\[-2ex]
& \multicolumn{6}{c|}{\underline{Confidence interval average length}} \\[1ex]
Full-info	&	0.179	&	0.002	&	0.010	&	0.162	&	0.002	&	0.009	\\
Independence	&	0.627	&	0.002	&	0.039	&	0.390	&	0.002	&	0.025	\\
Worst case	&	0.878	&	0.003	&	0.059	&	0.571	&	0.003	&	0.041	\\

& \multicolumn{6}{c|}{ } \\[-2ex]
& \multicolumn{6}{c|}{\underline{RMSE relative to true parameter values}} \\[1ex]
Full-info	&	1.53\%	&	0.59\%	&	0.86\%	&	1.37\%	&	0.58\%	&	0.76\%	\\
Independence	&	1.53\%	&	0.59\%	&	0.86\%	&	1.72\%	&	0.59\%	&	0.99\%	\\
Worst case	&	1.53\%	&	0.59\%	&	0.86\%	&	1.79\%	&	0.59\%	&	1.03\%	\\

& \multicolumn{6}{c|}{ } \\[-2ex]
& \multicolumn{6}{c|}{\underline{Rejection rate of over-identification test}} \\[1ex]
Full-info & \multicolumn{3}{c|}{5.01\%} & \multicolumn{3}{c|}{ } \\
Independence &  \multicolumn{3}{c|}{0.00\%} & \multicolumn{3}{c|}{ } \\
Worst case & \multicolumn{3}{c|}{0.00\%} & \multicolumn{3}{c|}{ } \\

& \multicolumn{6}{c|}{ } \\[-2ex]
& \multicolumn{6}{c|}{\underline{Rejection rate of joint test of true parameter values}} \\[1ex]
Full-info & \multicolumn{3}{c|}{4.79\%} & \multicolumn{3}{c|}{ } \\
Independence &  \multicolumn{3}{c|}{7.54\%} & \multicolumn{3}{c|}{ } \\
Worst case & \multicolumn{3}{c|}{2.47\%} & \multicolumn{3}{c|}{ } \\

\hline
\end{tabular}
\caption{Simulation results based on the empirically calibrated \citet{Alvarez2014} model. The just-identified specification uses only three moments for estimation, while the efficient specification exploits all four moments. The rows correspond to full-information inference (exploits knowledge of $\hat{V}$), inference under independence (erroneously assumes that $\hat{V}$ is diagonal), and worst-case inference (exploits only diagonal of $\hat{V}$ without assuming off-diagonal elements are zero). Estimated parameters: number of products ({\#} prod.), volatility of desired log price (Vol.), scaled menu cost (Menu cost). The over-identification test tests the validity of the fourth non-targeted moment. The joint test of the true parameters is a Wald test (Full-info or Independence) or the test proposed in \cref{sec:test_joint} (Worst case). The nominal significance level is 5\%.}
\label{tab:sim}
\end{table}

\cref{tab:sim} shows that the just-identified and efficient limited-information confidence intervals have coverage probabilities very nearly equal to or exceeding the nominal level of 95\% for all three parameters. Though coverage is conservative, the table shows that the average length of the confidence intervals is not more than six times that of the corresponding full-information confidence intervals. This is consistent with the empirical standard errors reported in \cref{sec:appl_price}.

\cref{tab:sim} also illustrates that the worst-case perspective is key to avoiding over-rejection in the face of limited information: Both the ``efficient'' t-test for the price volatility parameter and the joint Wald test of the true parameter values over-reject if we erroneously assume that all the empirical moments are independent of each other.\footnote{In the case of the volatility parameter, the independence-based efficient estimator loads positively on all four moments (i.e., $\hat{x}_j>0$ for all $j$). Because three of these moments are in fact highly positively correlated with each other (see \cref{tab:appl_price_moments}), ignoring covariances leads to an underestimate of the standard error.} In contrast, the limited-information and full-information t-tests and joint tests are correctly sized.\footnote{We only report the joint test for the just-identified specification. This is because the joint test proposed in \cref{sec:test_overid} requires a single choice of weight matrix, whereas the worst-case efficient point estimates of the three parameters correspond to three different choices of moments (selected as in \cref{sec:weight}).} The limited-information tests are conservative (as predicted by theory), though the joint test of parameter restrictions is only mildly conservative in this particular model.

Finally, \cref{tab:sim} shows that the limited-information efficient point estimates have slightly higher root mean squared error (RMSE) than the efficient full-information estimates. It may seem surprising at first blush that the limited-information efficient estimates can have (marginally) higher RMSE than the just-identified estimates. This is because the limited-information efficient estimates are designed to have low variance under the \emph{worst-case} correlation structure (i.e., perfect correlation of the moments), not under the true correlation structure that is unknown to the econometrician.

\subsection{Details of the heterogeneous agent application}
\label{sec:appendix_appl_hank}
We here provide further details on the application in \cref{sec:appl_hank}. \cref{tab:appl_hank_load} shows which impulse response moments are used to efficiently estimate the seven structural parameters, according to the moment selection procedure described in \cref{sec:weight}. The $p=23$ moments are shown along the rows, while the $k=7$ parameters are shown along the columns. A cell with an ``x'' indicates a non-zero efficient loading ($\hat{x}_j^*$ in the notation of \cref{sec:weight}), while empty cells indicate zero loadings.\footnote{We define a loading to be zero if $|\hat{x}_j^*|<10^{-4}$.}

\begin{table}[p]
\centering
\textsc{Heterogeneous agent application: Efficient moment selection} \\[0.5\baselineskip]
\renewcommand{\arraystretch}{1.2}
\begin{tabular}{|ccc||c|c|ccc|cc|}
\hline
\multicolumn{3}{|c||}{Impulse response} & & & \multicolumn{3}{c|}{TFP} & \multicolumn{2}{c|}{Monetary} \\
\cline{1-3} \cline{6-10}
Var. & Shock & Horiz. & TR & PC & AR1 & AR2 & Std & AR1 & AR2 \\
 \hline\hline
TFP & TFP & 0	&	x	&		&	x	&	x	&	x	&	x	&		\\
 &	& 1&	x	&		&	x	&	x	&		&	x	&		\\
 &	& 2&	x	&	x	&		&	x	&		&	x	&		\\
 &	& 8&		&		&		&		&		&		&		\\
 
\hline
Output & TFP	& 0	&	x	&	x	&		&		&		&	x	&		\\
&	& 1	&		&		&		&		&		&		&		\\
&	& 2	&		&		&		&		&		&		&		\\
&	& 8	&		&		&		&		&		&		&		\\
\hline
Frac & TFP 	& 0	&		&		&		&		&		&		&		\\
&	& 1	&		&		&		&		&		&		&		\\
&	& 2	&		&		&		&		&		&		&		\\
&	& 8	&		&		&		&		&		&		&		\\
\hline
Output & MP	& 0	&	x	&	x	&		&		&		&	x	&	x	\\
&	& 1	&		&		&		&		&		&		&		\\
&	& 2	&		&		&		&		&		&		&		\\
&	& 8	&		&		&		&		&		&		&		\\
\hline
Price & MP	& 0	&		&	x	&		&		&		&		&	x	\\
&	& 1	&		&		&		&		&		&		&		\\
&	& 2	&		&		&		&		&		&	x	&		\\
&	& 8	&		&		&		&		&		&		&		\\
\hline
Bond & MP	& 1	&	x	&	x	&		&		&		&	x	&	x	\\
&	& 2	&		&		&		&		&		&		&		\\
&	& 8	&	x	&		&		&		&		&		&	x	\\
\hline
\end{tabular}
\caption{Cells with an ``x'' indicate that the efficient estimate of the given parameter (along columns) attaches a non-zero weight to the given empirical moment (along rows). First three columns show the impulse response variable (``Var.''), shock, and quarterly horizon (``Horiz.''). Variable ``Frac'': fraction of people earning less than 2/3 of GDP. Shock ``MP'': monetary shock. See parameter abbreviations in \cref{tab:appl_hank_estim}.} \label{tab:appl_hank_load}
\end{table}

\clearpage
\phantomsection
\addcontentsline{toc}{section}{References}
\bibliography{calibration_ref}

\end{document}


\title{Online Supplement to \texorpdfstring{\\}{}``Standard Errors for Calibrated Parameters''}
\author{Matthew D. Cocci \and Mikkel Plagborg-M{\o}ller}
\date{June 17, 2024}
\maketitle

\noindent In Section 5.2 of the main paper, we estimate parameters in a one-asset HANK model.
This appendix gives further details on that model and the estimation
exercise we perform.
The model we estimate is nearly
identical to \citet{Auclert2020} as described in their Appendix B.2 and
estimated in their Appendix F.2.\footnote{%
  That model is itself adapted from \cite{McKay2016}.
}
The purpose of this document is to restate the full model of
\citet{Auclert2020} and define the small number of departures from that
baseline that we adopt for our empirical application.

Note that, in order to stay as close to \citet{Auclert2020} as possible, the notation in this online supplement conflicts with some of the notation in Sections 2--4 of our main paper.

\section{Model summary and empirical specification}

This section summarizes the equilibrium conditions of the one-asset HANK
model in Appendix B of \citet{Auclert2020} that we adopt for the application in Section 5.2 of our main paper. We explain in detail the minor changes that we make to the specification of the Taylor rule and exogenous shock processes relative to \citet{Auclert2020}. Finally, we specify the parameters that we estimate. For ease of access, \cref{sec:full} below reviews the full micro foundations for the model, as also explained in \citet{Auclert2020}.

\paragraph{Equilibrium conditions.}
The one-asset HANK model in Appendix B of \citet{Auclert2020} consists
of a unit mass of heterogeneous households that spend, save, and borrow
(up to a limit), a unit mass of monopolistically competitive
intermediate goods firms, a competitive final goods market, a monetary
authority responsible for setting the nominal interest rate, and a
fiscal authority responsible for taxation and government spending.
The equilibrium conditions
for model aggregates are given by
\begin{align}
  F_t(\boldsymbol{X},\boldsymbol{Z})
  &=
  \begin{pmatrix}
    Y_t - Z_t N_t
    \\
    Y_t
    -
    \frac{\mu_t}{\mu_t-1} \frac{1}{2\kappa}
    \log(1+\pi_t)^2
    Y_t
    -
    w_t
    N_t
    -
    d_t
    \\
    r_tB
    +
    G_t
    -
    \tau_t
    \\
    r_t^*
    + \phi\pi_t
    + \phi_y(Y_t-Y_{ss})
    - i_t
    \\
    1 + r_t
    -
    \frac{1+i_{t-1}}{1+\pi_t}
    \\
    \kappa
    \left(
      \frac{w_t}{Z_t}
      -
      \frac{1}{\mu_t}
    \right)
    +
    \frac{1}{1+r_{t+1}}
    \frac{Y_{t+1}}{Y_t}
    \log(1+\pi_{t+1})
    -
    \log(1+\pi_{t})
    \\
    \mathcal{A}_t(\{r_s,w_s,\tau_s,d_s\})
    -
    B
    \\
    \mathcal{N}_t(\{r_s,w_s,\tau_s,d_s\})
    -
    N_t
  \end{pmatrix}
  =
  \begin{pmatrix}
    0
    \\
    0
    \\
    0
    \\
    0
    \\
    0
    \\
    0
    \\
    0
    \\
    0
  \end{pmatrix}
  ,
  \label{eqm}
\end{align}
where
$\mathcal{A}_t(\cdot)$
and
$\mathcal{N}_t(\cdot)$ denote the aggregate savings and labor supply across heterogeneous households given the interest rate,
wage rate, taxes, and dividends.
The equations represent (in order)
production,
dividends distributed,
the government budget constraint,
the Taylor rule,
the Fisher equation,
the Phillips curve,
asset market clearing,
and
labor market clearing.
The aggregate resource constraint is omitted by Walras' law.
For further details on the model economy and equilibrium conditions, see
either Appendix B of \citet{Auclert2020} or \cref{sec:full} of this
document.

In the above equation, $\boldsymbol{X}=(\boldsymbol{X}_t)_{t \in \mathbb{Z}}$ denotes the sequence of endogenous aggregate variables, the period-$t$ values of which are $\boldsymbol{X}_t=(Y_t,N_t,\pi_t,w_t,d_t,r_t,\tau_t,i_t)$, while $\boldsymbol{Z}=(\boldsymbol{Z}_t)_{t \in \mathbb{Z}}$ denotes the sequence of exogenous driving variables with $\boldsymbol{Z}_t=(Z_t,G_t,r_t^*,\mu_t)$.  We specify the dynamics of the latter below.

\paragraph{Steady state.}
The steady state is defined as the (constant) value $\boldsymbol{X}_{ss}$ for $\boldsymbol{X}_t$ implied by the exogenous variables being equal to the constant values $\boldsymbol{Z}_{ss}=(Z,G,r,\mu)$ at all points in time. It is easy to see that this zero-inflation steady state does not depend on the parameters $(\kappa,\phi,\phi_y)$. We follow \citet{Auclert2020} and assume that the values of the remaining model parameters are known and given by the values listed in \cref{tab:params}.\footnote{%
  Values are taken from their Table B.2.
} We solve for the steady state at these parameter values, using exactly the same code as in \citet{Auclert2020}.\footnote{%
  For brevity, we suppress details on the discretization of grids
  and dynamics for idiosyncratic states in the household problem.
}

\begin{table}[t]
\centering
\begin{tabular}{clc}
  \hline
  Parameter & & Value
  \\\hline\hline
  $\beta=(1+r)^{-1}$ & Discount factor & 0.982 \\
  $\varphi$ & Disutility of labor & 0.786 \\
  $\sigma$ & Inverse IES & 2 \\
  $\nu$ & Inverse Frisch & 2 \\
  $\mu$ & Steady-state markup & 1.2 \\
  $B$ & Bond supply & 5.6 \\
  $G$ & Steady state government spending & 0 \\
  $Z$ & Steady state TFP & 1 \\
  $\underline{a}$ & Borrowing limit & 0 \\
  $\rho_e$ & Autocorrelation of efficiency hours & 0.966 \\
  $\sigma_e$ & Efficiency hour shock std.\ dev. & $0.5\sqrt{1-\rho_e^2}$
  \\\hline
\end{tabular}
\caption{%
  Values for non-estimated parameters. Note that some of these parameters only enter into the equilibrium conditions \eqref{eqm} through the aggregate household decision functions $\mathcal{A}_t(\cdot)$ and $\mathcal{N}_t(\cdot)$.
}
\label{tab:params}
\end{table}

\paragraph{Exogenous disturbance processes.}
We now describe the processes we assume for the exogenous variables $\boldsymbol{Z}_t=(Z_t,G_t,r_t^*,\mu_t)$. Our specification generalizes that of \citet{Auclert2020}, as discussed further below. Define log TFP growth $z_t = \log(Z_t/Z_{t-1})$, government spending in deviation from steady state as a fraction of steady state output $g_t=(G_t-G)/Y_{ss}$, and the markup and natural rate in deviation from steady state, $\tilde{\mu}_t = \mu_t-\mu$ and $\tilde{r}_t^* = r_t^*-r$. We assume the following exogenous AR(2) processes:
\begin{align*}
  z_t &= \rho_{z1} z_{t-1} + \rho_{z2} z_{t-2} + \varepsilon^z_t,
  \qquad
  \varepsilon^z_t
  \sim
  (0,\sigma^2_z),
  \\
  g_t &= \rho_{g1} g_{t-1} + \rho_{g2} g_{t-2} + \varepsilon^g_t,
  \qquad
  \varepsilon^g_t
  \sim
  (0,\sigma^2_g),
  \\
  \tilde{r}^*_t &= \rho_{r1} \tilde{r}_{t-1}^* + \rho_{r2} \tilde{r}_{t-2}^* + \varepsilon^{r}_t,
  \qquad
  \varepsilon^{r}_t
  \sim
  (0,\sigma^2_{r}),
  \\
  \tilde{\mu}_t &= \rho_{\mu 1} \tilde{\mu}_{t-1} + \rho_{\mu 2} \tilde{\mu}_{t-2} + \varepsilon^\mu_t,
  \qquad
  \varepsilon^\mu_t
  \sim
  (0,\sigma^2_\mu).
\end{align*}
All the shocks $\varepsilon_t^j$ for $j \in \lbrace z,g,r,\mu\rbrace$ are i.i.d.\ and mutually independent.

\paragraph{Modifications relative to the \citet{Auclert2020} specification.}
Our empirical specification differs from that in \citet[Table F.III]{Auclert2020} in the following ways.

First, the \citet{Auclert2020} model specification obtains if (i) we assume that there are no TFP shocks ($\sigma_z=0$) and (ii) we restrict the exogenous variables to follow AR(1) processes ($\rho_{g2}=\rho_{r2}=\rho_{\mu2}=0$). We drop all these assumptions and thereby strictly generalize their specification of the exogenous driving forces.

Second, while \citet{Auclert2020} estimate the Taylor rule parameter $\phi_y$ on output, we restrict this parameter to equal 0. This restriction is in line with the numerically small value of $\phi_y$ estimated by \citet[Table F.III]{Auclert2020}. The restriction is also imposed in the variant of the one-asset HANK model featured in the online Python toolbox developed by \citeauthor{Auclert2020}\footnote{\url{https://github.com/shade-econ/sequence-jacobian}}

In summary, the minor modifications we make relative to \citet{Auclert2020} strictly generalize their assumptions on the exogenous driving processes, but impose the restriction that the Taylor rule depends only on inflation and not output. These are the only differences in model specification between our application and that of \citet[Table F.III]{Auclert2020}. Recall that these changes do not affect the model's steady state.

\paragraph{Estimated parameters.}
The 7 parameters we estimate are: the Taylor rule coefficient on inflation $\phi$, the slope of the Phillips curve $\kappa$, the parameters of the autoregressive process for TFP growth $(\rho_{z1},\rho_{z2},\sigma_z)$, and the autoregressive coefficients for the monetary shock $(\rho_{r1},\rho_{r2})$. We do not explicitly estimate the parameters in the autoregressive processes for the government spending disturbance $g_t$ and the markup disturbance $\tilde{\mu}_t$, as these parameters do not influence the particular impulse responses that we match (see the definition of the matched impulse responses in our main paper). Note, however, that our estimation procedure allows for any values of these non-estimated parameters, as explained in our main paper.\footnote{In fact, we do not even need to restrict these processes to be autoregressions of order 2.} We similarly do not estimate the standard deviation of the monetary shock $\sigma_r$, since this parameter does not affect the \emph{normalized} impulse responses with respect to monetary shocks that we match (see the definition of the matched impulse responses in our main paper). Again, we note that our estimation procedure allows for any value of $\sigma_r$, though this parameter is not directly estimated.

\citet{Auclert2020} estimate all parameters of the government spending, markup, and monetary disturbance processes, though they restrict these processes to be autoregressions of order 1 rather than order 2. The fact that they estimate these parameters, however, does not mean that they allow for more shocks in their model than we do. Our impulse response matching procedure is robust to the presence of other shocks, including the government spending and markup shocks explicitly defined above.

\section{Detailed model specification}
\label{sec:full}

This section gives a full accounting of the microfoundations of the one-asset HANK model, the equilibrium conditions of which were summarized above.
This section merely reviews the model of Appendix B.2 of \citet{Auclert2020},
whose modeling assumptions we adopt without any changes (recall that the only changes we make to the empirical specification are to the exogenous disturbance processes, as defined earlier).

\paragraph{Households.}
Each heterogeneous household's problem is characterized by the following
Bellman equation:
\begin{align*}
  V_t(e_{it},a_{it-1})
  =
  &
  \max_{c_{it},n_{it},a_{it}}
  \left\{
    \frac{c_{it}^{1-\sigma}}{1-\sigma}
    -
    \varphi
    \frac{%
      n_{it}^{1+\nu}
    }{%
      1+\nu
    }
    +
    \beta
    \mathbb{E}_t
    V_{t+1}
    (e_{it+1},a_{it})
  \right\}
  \\
  &
  c_{it}
  +
  a_{it}
  =
  (1+r_t)a_{it-1}
  +
  w_te_{it}n_{it}
  -
  \tau_t
  \bar{\tau}(e_{it})
  +
  d_t
  \bar{d}(e_{it})
  \\
  &
  a_{it}\geq \underline{a}
\end{align*}
Specifically, each household optimally chooses its real consumption
$c_{it}$, labor supply $n_{it}$, and real asset holdings
$a_{it}$ subject to a budget and borrowing constraint $a_{it}\geq
\underline{a}$.
Income for consumption and savings is generated from three sources.
First, households can supply labor $n_{it}$, for which they earn a wage
rate $w_t$ per efficiency hour $e_{it}$.
Households' idiosyncratic efficiency state $e_{it}$ evolves exogenously 
over time, and the probability of transitioning from state $e$ to state
$e'$ is governed by a joint probability distribution denoted by
$P(e,e')$.
This transition kernel comes from a discretization of an AR(1) process for $\log e_{it} = \rho_e \log e_{i,t-1} + \epsilon_{i,t}$, $\epsilon_{i,t} \sim N(0,\sigma_e^2)$.
Second, households have access to (real) principal and interest
$(1+r_t)a_{it-1}$ given their previous period asset holdings.
Third, households receive dividend income $d_t\overline{d}(e_{it})$ from
owning monopolistically competitive intermediate goods firms, to be
discussed below.
Subtracting from income, households must pay taxes
$\tau_t\overline{\tau}(e_{it})$ in each period.
Note that
$\overline{d}(\;\cdot\;)$
and
$\overline{\tau}(\;\cdot\;)$
are densities integrating to one;
therefore,
$d_t$
and
$\tau_t$
control the overall level.

\paragraph{Firms.}
First, aggregate output is produced by a single representative final
goods firm in a competitive market that aggregates intermediate goods
$y_{jt}$ produced by a unit mass of monopolistically competitive
intermediate goods firms according to
\begin{align*}
  Y_t
  =
  \left(
      \int_0^1
      y_{jt}^{1/{\mu_t}}\;dj
  \right)^{\mu_t}
  ,
\end{align*}
where $\mu_t$ is an exogenous markup process. Let $P_t$ denote the aggregate price level for final goods (more on this
below).
Let $p_{jt}$ denote the price of the $j$th intermediate good.
Taking the price of final goods $P_t$ and demand $Y_t$ as
given,\footnote{%
  These will be pinned down later by market clearing and zero profits.
  At this point, the task is only to determine how $Y_t$ is created from
  intermediate goods.
}
the final good firm determines how to procure specific intermediate
goods by solving the (nominal) profit maximization problem
\begin{align}
  \max_{y_{jt}}
  P_t
  Y_t
  -
  \int_0^1
  p_{jt}y_{jt}
  \;dj
  \qquad
  \text{s.t.}
  \quad
  Y_t
  =
  \left(
      \int_0^1
      y_{jt}^{1/{\mu_t}}\;dj
  \right)^{\mu_t}
  .
  \label{final_profitmax}
\end{align}
Next, each intermediate goods firm $j$ produces $y_{jt}$ according to
production function
\begin{align*}
  F(n_{jt}) = Z_t n_{jt}
  ,
\end{align*}
where $Z_t$ is aggregate TFP and $n_{jt}$ is the
efficiency units of labor hired by intermediate firm $j$.\footnote{%
  This is a distinct definition from a household's $n_{it}$ supply.
  Whereas firms' $n_{jt}$ is a measure of efficiency units hired,
  $n_{it}$ is raw hours, before scaling by efficiency.
  We handle this distinction in market clearing.
}
Therefore, to produce $y_{jt}$ units of intermediate output, the firm
must hire $n_{jt}=y_{jt}/Z_t$ efficiency units of labor.
Each intermediate firm can set their price in period $t$ subject to a
quadratic (real) adjustment cost
\begin{align*}
  \psi_t
  (p_{jt},p_{jt-1})
  &=
  \frac{\mu_t}{\mu_t-1}
  \frac{1}{2\kappa}
  \log
  \left(
    \frac{p_{jt}}{p_{jt-1}}
  \right)^2
  Y_t
  .
\end{align*}
Under perfect foresight, taking as given demand schedules
$y_{jt}(p_{jt})$ for intermediate good $j$,
firm $j$ chooses prices to maximize discounted (real) profits
\begin{align}
  &
  \max_{\{p_{j,t+s}\}_{s=0}^\infty}
  \sum_{s=0}^\infty
  M_{t,t+s}
  \bigg\{
    \frac{p_{j,t+s}}{P_{t+s}}
    y_{j,t+s}(p_{jt})
    -
    w_{t+s}n_{j,t+s}
    \frac{%
      y_{j,t+s}(p_{jt})
    }{%
      Z_{t+s}
    }
    -
    \frac{\mu_{t+s}}{\mu_{t+s}-1}
    \frac{1}{2\kappa}
    \log
    \left(
      \frac{p_{j,t+s}}{p_{j,t+s-1}}
    \right)^2
    Y_{t+s}
  \bigg\},
  \label{objfcn_pricesetting}
\end{align}
where $M_{t,t+s}$ is the real discount factor, which comes from the
household problem, since intermediate goods firms are owned by
households.

\paragraph{Policy.}
The monetary authority sets the nominal interest rate $i_t$ according to
a Taylor rule
\begin{align}
  i_t
  =
  r_t^* + \phi \pi_t
  + \phi_y (Y_t - Y_{ss}),
  \label{taylor}
\end{align}
where $r_t^*$ is the exogenous natural interest rate with steady state value $r = \beta^{-1}-1$,
$1+\pi_t=P_t/P_{t-1}$ is inflation in the nominal price index, and
$Y_{ss}$ is steady state aggregate output.
This rule, together with inflation, determines the real interest rate
$r_t$ via the Fisher equation
\begin{align}
  1 + r_t = \frac{1+i_{t-1}}{1+\pi_t}
  .
  \label{fisher}
\end{align}
Explicitly, the real rate of interest paid out in period $t$ on the
asset purchased in period $t-1$ equals the promised nominal rate at time
$i_{t-1}$ less realized inflation $\pi_t$ between $t-1$ and $t$.

Finally, the fiscal authority has exogenous spending $G_t$ and pays interest at rate
$r_t$ on the outstanding stock of bonds with constant real face value $B$.

\paragraph{Equilibrium.}
Equilibrium is characterized by (1) a sequence of policy functions and
value functions for the household problem
$\{c_t(e,a_{-1}),n_t(e,a_{-1}),a_t(e,a_{-1}),V_t(e,a_{-1})\}_{t=0}^\infty$,
(2)
a sequence of distributions $\{\Gamma_t(e,a)\}_{t=0}^\infty$ over the
idiosyncratic states with support $\mathcal{E}\times \mathcal{A}$,
(3) and aggregate sequences
$\{Y_t,N_t,\pi_t,w_t,d_t,r_t,\tau_t,i_t,\psi_t\}_{t=0}^\infty$
consistent with equilibrium conditions, given the exogenous processes. We now state these equilibrium conditions.

First, bond market clearing requires that aggregate savings by
households equals the stock of outstanding bonds in each period,
\begin{align}
  B
  =
  \int_{\mathcal{E}}
  \int_{\mathcal{A}}
  a_t(e,a_{-1})
  d\Gamma_t(e,a_{-1})
  .
  \label{assets}
\end{align}
The government must also balance its budget in each period,
\begin{align}
  r_t B
  +
  G_t
  =
  \int_{\mathcal{E}}
  \int_{\mathcal{A}}
  \tau_t
  \bar{\tau}(e)
  \,
  d\Gamma_t(e,a)
  =
  {\tau}_t
  ,
  \label{govt}
\end{align}
where $\tau_t$ is aggregate taxes.

On the firm side,
by standard arguments, there are zero profits in the competitive final
goods market and there is a symmetric equilibrium in the intermediate
goods market in which all monopolistically competitive intermediate
goods firms set identical prices, hire an identical amount of labor, and
produce an identical amount of output.
For brevity, we focus only upon the aggregate implications of
equilibrium in the production side of the model and suppress details at
the level of individual firms.
Therefore,
let $P_t$ denote the price level in period $t$,
let $N_t$ denote the aggregate efficiency units of labor hired across
all firms, and let $Y_t$ denote the corresponding amount of the final
good produced.
Equilibrium in the intermediate goods market implies the following
Phillips curve
\begin{align}
  \log
  \left(
    1+\pi_t
  \right)
  &=
  \kappa
  \left[
  \frac{w_t}{Z_t}
  -
  \frac{1}{\mu_t}
  \right]
  +
  \frac{1}{1+r_{t+1}}
  \log
  \left(
    1+\pi_{t+1}
  \right)
  \frac{Y_{t+1}}{Y_t}
  .
  \label{phillips}
\end{align}
Second, real aggregate output of the final good must equal aggregate
production
\begin{align}
  Y_t = Z_tN_t
  .
  \label{production}
\end{align}
Lastly on the firm side, aggregate real dividends remitted to households
equals total real output less labor and price adjustment costs,
\begin{align}
  \int_\mathcal{E}
  \int_\mathcal{A}
  d_t\bar{d}(e)
  \,
  d\Gamma(e,a)
  =
  d_t
  &=
  Y_t
  -
  w_tN_t
  -
  \frac{\mu_t}{\mu_t-1}
  \frac{1}{2\kappa}
  \log
  \left(
    1+\pi_t
  \right)^2
  Y_t
  .
  \label{dividends}
\end{align}
Next,
labor market clearing requires labor demand from firms
equals households' supply,
\begin{align}
  N_t
  :=
  \int_{\mathcal{E}}
  \int_{\mathcal{A}}
  e
  \,
  n_t(e,a_{-1})
  \,
  d\Gamma_t(e,a_{-1})
  .
  \label{labor}
\end{align}
The aggregate resource constraint requires that private consumption
equals total output less government spending and aggregate price
adjustment costs,
\begin{align}
  \int_{\mathcal{E}}
  \int_{\mathcal{A}}
  c_t(e,a_{-1})
  d\Gamma_t(e,a_{-1})
  =
  Y_t
  -
  G_t
  -
  \frac{\mu_t}{\mu_t-1}
  \frac{1}{2\kappa}
  \log
  \left(
    1+\pi_t
  \right)^2
  .
  \label{resource}
\end{align}
This last equation can be derived by aggregating the budget constraint
of the individual heterogeneous households and using
Equations~\eqref{assets}, \eqref{govt}, \eqref{dividends}, and \eqref{labor}
to simplify.

Finally, the distributions must be consistent with the policy functions
of households,
\begin{align*}
  \Gamma_{t+1}
  (e',a)
  =
  \int_{\mathcal{E}}
  \int_{\mathcal{A}}
  \mathbf{1}\{a_t(e,a_{-1})=a\}
  P(e'|e)
  d\Gamma_{t}(e,a_{-1})
  .
\end{align*}

\paragraph{Sequence form.}
Stacking the equilibrium conditions in Equations \eqref{taylor}--\eqref{labor}, we obtain the system of equations \eqref{eqm} described in the previous section. The model is solved using the linearization technique developed by \citet{Auclert2020}.

\clearpage
\phantomsection
\addcontentsline{toc}{section}{References}
\bibliography{calibration_ref}